\documentclass[letterpaper,11pt]{article}
\usepackage{preamble_aditya_lattice}

\newcommand{\ETH}{\problem{ETH}}
\newcommand{\CVP}{\problem{CVP}}
\newcommand{\SVP}{\problem{SVP}}

\newcommand{\BDD}{\problem{BDD}}
\newcommand{\SAT}{\problem{SAT}}

\newcommand{\cl}{\ell}
\newcommand{\LIN}{\problem{MAXLIN}}

\newcommand{\di}{\operatorname{dist}}
\newcommand{\eps}{\varepsilon} 

\newcommand{\tdgr}{\vec{t}^{\dagger}}
\newcommand{\ldgr}{\dgr\lat}
\newcommand{\gmlin}{\ensuremath{\mathsf{gap}_{\frac{5}{8},\frac{5}{8}\textrm{-}\varepsilon}\LIN}}
\newcommand{\lin}{\problem{MAXLIN}}
\newcommand{\geth}{\problem{GapETH}}
\newcommand{\sat}{\problem{3SAT}}
\newcommand{\nump}[1]{\ensuremath{{N}_p\paren{#1}}}

\newcommand{\ballp}[1]{\ensuremath{\mathcal{B}_p\paren{#1}}}

\newcommand{\cL}{\mathcal{L}}
\newcommand{\normpp}[1]{\norm{#1}_p^p}
\newcommand{\dgr}[1]{#1^{\dagger}}

\usepackage{thmtools}
\usepackage{thm-restate}
\newcommand{\svpalphaval}{\frac{2r'^p}{\delta r^{\dagger p}}}
\newcommand{\svprval}{ 1 + \paren{1 - \frac{\delta}{2}} \frac{2r'^p}{\delta}}
\newcommand{\svpgammaval}{ 1 + \delta  \min \left\{ \frac{(\sqrt{\gamma'^p}-1)^2}{2}, \frac{1}{100} \right\}}

\begin{document}

\title{Mind the Gap? Not for SVP Hardness under ETH!}
	\author{Divesh Aggarwal\thanks{National University of Singapore. Email: \texttt{divesh@comp.nus.edu.sg}} 
		\and 
		Rishav Gupta\thanks{National University of Singapore. Email: \texttt{rishavg@u.nus.edu}}
        \and
        Aditya Morolia\thanks{Centre for Quantum Technologies, Singapore. Email: \texttt{morolia@u.nus.edu}}
		\and 
		Chuanqi Zhang\thanks{Monash University. Email: \texttt{chuanqi.zhang@monash.edu}}
        }





	\date{}
	\maketitle

\begin{abstract}
We prove new hardness results for fundamental lattice problems under the Exponential Time Hypothesis (ETH). Building on a recent breakthrough by Bitansky et al.\ [BHIRW24], who gave a polynomial-time reduction from $\problem{3SAT}$ to the (gap) $\lin$ problem—a class of CSPs with linear equations over finite fields—we derive ETH hardness for several lattice problems. 

First, we show that for any $p \in [1, \infty)$, there exists an explicit constant $\gamma > 1$ such that $\cvp_{p,\gamma}$ (the $\ell_p$-norm approximate Closest Vector Problem) does not admit a $2^{o(n)}$-time algorithm unless ETH is false. Our reduction is deterministic and proceeds via a direct reduction from (gap) $\lin$ to $\cvp_{p,\gamma}$.

Our main contribution is a randomized ETH hardness result for $\svp_{p,\gamma}$ (the $\ell_p$-norm approximate Shortest Vector Problem) for all $p \in (2, \infty)$. This result relies on a novel geometric property of the integer lattice $\Z^n$ in the $\ell_p$ norm, which says that for any $p \in (2, \infty)$, the number of lattice vectors close to $\frac{1}{2}\vec 1_n$ (in the $\ell_p$ norm) is exponentially larger than the number of short vectors (namely those close to the origin). We establish this property via a new inequality for the Theta function, which we use to get a randomized reduction from $\cvp_{p,\gamma}$ to $\svp_{p,\gamma'}$.

Finally, we also use our ideas to give some minor improvements over prior reductions from $\problem{3SAT}$ to $\bdd_{p, \alpha}$ (the Bounded Distance Decoding Problem), yielding better ETH hardness results for $\bdd_{p, \alpha}$ for any $p \in [1, \infty)$ and $\alpha > \alpha_p^{\ddagger}$, where $\alpha_p^{\ddagger}$ is an explicit threshold depending on $p$. 

\end{abstract}
\newpage
{\setlength{\parskip}{2pt}
    \tableofcontents}
	\setcounter{page}{0}
	\thispagestyle{empty}
	\newpage
	\setcounter{page}{1}
	
	\section{Introduction}
	\label{sec:intro}
    	A lattice in $d$ dimensions is a discrete subgroup of $\R^d$. Formally, given a set of vectors $\vec v_1, \ldots, \vec v_n \in \Q^d$, the lattice generated by these vectors is defined as
\begin{equation*}
    \lat = \lat(\vec v_1, \ldots, \vec v_n) := \left\{\sum_{i=1}^n a_i \vec v_i : a_i \in \Z \right\}.
\end{equation*}

Several algorithmic problems on lattices are of particular importance. The \textbf{Shortest Vector Problem} (\textsf{SVP}) asks for the shortest nonzero vector in a given lattice—typically measured in an $\ell_p$ norm. The \textbf{Closest Vector Problem} (\textsf{CVP}) asks for the lattice vector closest to a given target vector in the ambient space $\R^d$. A related variant is the \textbf{Bounded Distance Decoding} (\textsf{BDD}) problem, which can be viewed as \textsf{CVP} under the promise that the target point lies sufficiently close to the lattice. 

Algorithms for solving these problems have led to impactful applications across multiple domains. These include polynomial factoring~\cite{LLL82}, integer programming~\cite{Lenstra83,Kannan87,DPV11,RR23}, and cryptanalysis~\cite{Shamir84,Odl90,JS98,NS01}. At the same time, the conjectured hardness of lattice problems has enabled the design of powerful cryptographic primitives—particularly in the context of \emph{post-quantum cryptography}. These schemes are not only conjectured to be secure against quantum adversaries but also support advanced functionalities such as fully homomorphic encryption~\cite{oded05,GPV08,BV11,new_hope}. Remarkably, many of these constructions rely on the \emph{worst-case hardness} of lattice problems, in contrast to traditional schemes based on the \emph{average-case hardness} of factoring or discrete logarithm~\cite{Ajtai96,oded05}. 

The security of most finalists in the NIST post-quantum cryptography standardization process rests on the assumed difficulty of solving lattice problems~\cite{NIST_quantum}.

\textbf{Algorithmic Progress on Lattice Problems.}
The algorithmic study of $\CVP$ and $\SVP$ in the Euclidean norm ($\ell_2$) began with the LLL algorithm~\cite{LLL82}, which provides a $2^{O(n)}$-approximation, and continued through foundational works of Babai~\cite{Bab86}, Kannan~\cite{Kannan87}, Schnorr~\cite{Sch87}, and Ajtai, Kumar, and Sivakumar~\cite{AKS01}. The fastest known algorithms for exact $\SVP$ are based on randomized sieving~\cite{AKS01}, extended to all $\ell_p$ norms~\cite{BN09}, general norms~\cite{AJ08}, and even asymmetric convex bodies~\cite{DPV11}, yielding $2^{O(n)} \cdot \poly(d)$-time algorithms for $\SVP_p$. In the case $p = 2$, a long line of work~\cite{AKS01,PS09,MV10} culminated in a $2^{n+o(n)}$-time algorithm~\cite{ADRS15,ADS15,AS17} for both exact $\SVP$ and exact $\CVP$, a $2^{0.835n}$ quantum algorithm for $\SVP$~\cite{ACKS25},  while constant-factor approximation for $\SVP$ is achievable in $2^{n/2+o(n)}$ time~\cite{ADRS15}. Even faster heuristic sieving algorithms are known under plausible assumptions~\cite{NguyenVidick08,WLTB11,Laarhoven2015,BDGL16}. For polynomial approximation factors, the best known algorithms for $\SVP_2$ run in time $2^{Cn}$, where the constant $C$ depends critically on the approximation factor~\cite{GN08,ALNS20,ALS21}, and small improvements in $C$ have significant implications for cryptographic security.
Recently, Eisenbrand and Venzin~\cite{EV20} showed that the fastest constant approximation factor $\SVP_2$ algorithm can be adapted to solve constant approximation factor $\SVP_p$ and $\CVP_p$ for any $p$ with essentially the same runtime. Building on this~\cite{ACKLS21} developed tight, rank- and dimension-preserving reductions between $\SVP_p$ and $\CVP_p$ across all norms, handling both constant and polynomial approximation factors, thereby expanding the algorithmic toolkit for lattice problems in general $\ell_p$ spaces.

\textbf{Computational Hardness of Lattice Problems.}
A series of efforts \cite{Boas81,ABSS93,CN98,Mic01svp,Khot05svp,RR06,HR14} have shown that $\SVP_p$ and $\CVP_p$ are $\np$-hard to approximate to within any constant factor, and hard to approximate to within $n^{c/\log\log n}$ for a constant $c > 0$, under reasonable complexity theoretic assumptions. See \cite{Bennett23} for a recent survey on the hardness of $\svp$. These results however do not rule out the existence of sub-exponential time algorithms for $\SVP$. This question is of immense interest from a theoretical point of view, as well as from a practical (cryptographic) point of view. For instance, to break the minimally secure post-quantum cryptographic schemes currently being standardized (for example, \cite{Dilithium2021Algorithm}), one would need to solve $\SVP_2$ in roughly $400$ dimensions. At this stage, a $2^{n/\log n}$ time algorithm would be sufficient to break these schemes in practice. We need stronger and more {\em fine-grained} hardness assumptions to rule out the existence of such algorithms.

\textbf{Fine Grained Complexity.}
The \pcp theorem \cite{AS98,ALMSS98} was a breakthrough result that states that every language in \np has a polynomial-size proof that can be verified by a probabilistic verifier that reads only a constant number of bits from the proof. This implies \np hardness of approximation for problems such as $\problem{MAX3SAT}$ and $\lin$ \cite{Has01}. 
This rules out polynomial time algorithms for these problems (unless $\p = \np$), but does not provide any guarantees about the existence of sub-exponential time algorithms, which are of interest especially in cryptography. 
The analogous (to $\p \neq \np$) assumption to start from would be the Exponential Time Hypothesis (ETH) \cite{IP01}, which states that $\problem{3SAT}$ on $n$ variables cannot be solved in time $2^{o(n)}$. 
However, it is not known whether this can be used to rule out the existence of sub-exponential time approximation algorithms for problems such as $\problem{MAX3SAT}$ or $\lin$.
More recently, \cite{Din16,MR17} formulated this into a yet stronger hardness assumption called the \geth which states that {\em approximating} $\problem{MAX3SAT}$ on $n$ variables requires $2^{\Omega(n)}$ time.

\paragraph{Gap-ETH vs ETH.}

The Exponential Time Hypothesis (ETH) and its stronger variant, Gap-ETH, have both played central roles in establishing fine-grained complexity and inapproximability results. Gap-ETH, introduced in \cite{Din16}, states that for some constant $0<\eta < 1$, no algorithm can distinguish, in $2^{o(n)}$ time, between a $3\text{SAT}$ instance in which all $m$ clauses are satisfiable and one in which no assignment satisfies more than an $\eta$-fraction of the clauses. This stronger assumption has enabled sharp inapproximability results for a wide range of problems, including $2\text{CSP}$ \cite{Din16}, Densest $k$-Subgraph and $k$-Biclique \cite{densestfpt}, parameterized $\text{SVP}$ \cite{para_svp}, and $\text{TSP}$ \cite{tsp}. More recently, however, breakthrough works have shown how to derive similar gap-producing reductions under the weaker ETH assumption, leading to inapproximability results for problems such as $2\text{CSP}$, Gap $k$-Clique \cite{ethpih}, $\text{gapMAXLIN}$ \cite{BHIRW24}, and parameterized $\text{SVP}$ \cite{ethparsvp}. These results raise the intriguing possibility that ETH may in fact imply Gap-ETH—a connection that could be established, for instance, via a PCP for $3\text{SAT}$ with linear proof blowup, though the existence of such PCPs remains a long-standing open question. In light of this, a promising intermediate goal is to demonstrate that hardness results previously known only under Gap-ETH can in fact be obtained under ETH alone. Our work takes this approach: we use the gap-producing reduction from \cite{BHIRW24} to show the ETH hardness for the approximate lattice problems $\text{SVP}_{p,\gamma}$, $\text{CVP}_{p,\gamma}$, and $\text{BDD}_{p, \alpha}$, whose hardness was previously established only under Gap-ETH. Prior ETH-based hardness for these fundamental problems was unknown, though a sequence of works \cite{BGS17, AS18, AC21, bdd21, ABGS21} established their hardness under the stronger Gap-ETH assumption. Our results thus provide new evidence for the power of ETH and take a step toward bridging the gap between ETH and Gap-ETH in the context of lattice-based inapproximability.

\subsection{Our Results}

We study the hardness of lattice problems under the Exponential Time Hypothesis. Our results are summarized in \cref{tab:complexity_summary}.

At the heart of our results, is a recent breakthrough result of Bitansky et. al. \cite{BHIRW24} which shows that there is a polynomial time reduction from $\problem{3SAT}$ in, say, $n$ variables to gap version of $\lin$ problem with $\order{n}$ variables and $\order{n}$ clauses, which is a constraint satisfaction problem (CSP) where each clause is a linear equation over a finite field. It follows that (gap) $\lin$ is \ETH hard.

\renewcommand{\arraystretch}{1.2} 

\begin{table}[h!]
\begin{center}
		\begin{tabular}{|c|  c | c | c | c|}
			\hline
			
			Problem&  $\cl_p$-norm   & \textsf{Gap-ETH} & \textsf{ETH} &Notes \\
			\hline
			$\SVP_{p,\gamma}$ &$2 < p < \infty$ & ${ 2^{\Omega(n)}}$ & ${\color{blue4} 2^{\Omega(n)}}$ &   \\
			
			&$1 \leq p \leq 2$& --&-- &  \\
			&$p = \infty$ & ${2^{\Omega(n)}}$ & ${2^{\Omega(n)}}$& \cite{BGS17} \\
			\hline
            
			  $\CVP_{p,\gamma}$&$1 \leq p < \infty$ & ${2^{\Omega(n)}}$ & ${\color{blue4}2^{\Omega(n)}}$ &   \\
			
			&$p = \infty$&${2^{\Omega(n)}}$ & ${2^{\Omega(n)}}$ & \cite{BGS17} \\
			\hline
            
		\end{tabular}
    \end{center}
\caption{Summary of known fine-grained upper and lower bounds for $\SVP_{p,\gamma}$ and $\CVP_{p,\gamma}$ for various values of $p$ and some constant $\gamma>1$, under various assumptions, with our results in {\color{blue4} blue}.}
        \label{tab:complexity_summary}
\end{table}

\textbf{ETH Hardness of $\cvp$.} We define $\cvp_{p,\gamma}$ to be a decision version of the Closest Vector Problem, where given a matrix $B$ over, say, integers, a target vector $\vec t$, and a radius $r$, the goal is to decide if there exists a lattice vector $\vec v \in \lat(B)$ such that $\norm{\vec v - \vec t}_p \leq r$, or for all lattice vectors $\vec v$, $\norm{\vec v - \vec t}_p > \gamma r$. Here $\lat(B)$ is the lattice generated by the columns of $B$. 
Our first result says that for any $p \in [1, \infty)$, there is no sub-exponential (in the dimension of the lattice) algorithm for $\cvp_{p,\gamma}$ for some (explicit) constant $\gamma > 1$, unless the Exponential Time Hypothesis is false.
\begin{restatable}[\ETH hardness of $\cvp_{p, \gamma}$]{theorem}{thmcvpethhard}
\label{thm:CvpEthHard}
For any $p \in [1, \infty)$, there exists a constant $\gamma > 1$ such that for all $n \in \Z^+$, there is no $2^{o(n)}$ time algorithm for $\CVP_{p, \gamma}$ over $\R^n$, unless the Exponential Time Hypothesis is false. 
\end{restatable}
To prove this, we give a {\em deterministic} Karp reduction from $\lin$ over, say, $n$ variables and $m$ equations to $\cvp_{p,\gamma}$ in a lattice in $m$ dimensions.

\textbf{Randomized ETH Hardness of $\svp$.} We define $\svp_{p,\gamma}$ to be a decision version of the approximate shortest vector problem, where given a matrix $B$ over, say, integers, and a radius $r$, the goal is to decide if there exists a lattice vector $\vec v$ such that $\norm{\vec v}_p \leq r$, or if for all lattice vectors $\vec v$, we have that $\norm{\vec v}_p > \gamma r$.
Our main result is that for any $p > 2$, unless the randomized Exponential Time Hypothesis is false, there is no sub-exponential time algorithm for $\svp_{p,\gamma}$ for an explicit constant $\gamma > 1$.

\begin{restatable}[\ETH hardness of $\svp_{p, \gamma}$]{theorem}{thmsvpethhard}
    \label{thm:SvpEthHard}
    For any $p\in (2, \infty)$, there exists a constant $\gamma > 1$ such that for all sufficiently large $n \in \Z^+$, there is no $2^{o(n)}$ time algorithm for $\SVP_{p, \gamma}$ over $\R^n$ unless the randomized Exponential Time Hypothesis is false.
\end{restatable}

To prove this theorem we show a novel property of the integer lattice $\Z^n$, which says that for any $p > 2$ there are exponentially more lattice vectors close (in the $\ell_p$ norm) to $\frac{1}{2}\vec 1_n$, than the number of short vectors, i.e., vectors around the origin.
Then we borrow techniques from \cite{AS18} to show that we can use the $\cvp_{p, \gamma}$ instances constructed in the proof of Theorem \ref{thm:CvpEthHard} to generate an $\svp_{p, \gamma'}$ instance for some $\gamma' > 1$ via a Karp reduction. Chaining together the efficient reductions from $\sat$ to (gap) $\lin$, from (gap) $\lin$ to $\cvp_{p, \gamma}$ and from $\cvp_{p, \gamma}$ to $\svp_{p, \gamma'}$, we get our result. 

\textbf{$p = \infty$ case.} We note that the reduction from $k$-$\problem{SAT}$ to $\cvp_{\infty}$ and $\svp_{\infty}$ in \cite[Corollary 6.7]{BGS17} achieves a gap of $\gamma_p = 1 + 2 / (k-1)$. Therefore for $k = 3$, they get a gap of $2$, and hence we know the $\ETH$ hardness of $\cvp_{\infty, \gamma}$ and $\svp_{\infty, \gamma}$ for $\gamma = 2$ since their work. 

\textbf{A reduction from $\cvp$ to $\bdd$, and hence ETH Hardness of $\bdd$.} We define $\bdd_{p,\alpha}$ to be the following search problem. The input is a matrix $B$ over, say, integers and a target vector $\vec t$ under the promise that there exists a lattice vector $\vec v$ at a distance at most $\alpha \cdot \lambda_1^{(p)}(\lat(B))$, where $\lambda_1^{(p)}(\lat(B))$ is the length of the shortest non-zero vector in $\lat(B)$ in the $\ell_p$ norm. The goal is to find a lattice vector closest to $\vec t$.
We show that for any $p \in [1, \infty)$, there is an efficient decision-to-search  reduction from $\cvp_{p, \gamma}$ over a lattice over integers to $\bdd_{p, \alpha}$, for any constant $\gamma$ and $\alpha > \alpha_p^{\ddagger}$, where $\alpha_p^{\ddagger}$ is an explicit constant defined in \cref{eqn:alphapddagger}, such that $\alpha_p^{\ddagger} = 1$ for $p \in [1, 2]$, $\alpha_p^{\ddagger} < 1$ for $p > 2$, and $\alpha_p^{\ddagger} \to 1/2$ as $p \to \infty$.

\begin{restatable}[$\cvp_{p, \gamma}$ reduces to $\bdd_{p, \alpha}$]{theorem}{thmcvpbddreduction} 
\label{thm:CvpBddReduction}
    For any $\gamma'>1$, $c>0$, and $p\in [1,\infty)$, the following holds for all $\alpha>\alpha _p^{\ddagger}$ and sufficiently large $m \in \Z^+$.
    There is a decision-to-search reduction from any $\cvp_{p, \gamma'}$ to $\bdd_{p, \alpha}$, where the $\cvp_{p, \gamma'}$  instance $(B', \vec t', r')$ is such that $B' \in \Z^{m \times m'}$, $\vec t' \in \Z^m$ and $r' = c \cdot m^{1/p}$. 
\end{restatable}

As a consequence, combining this result with our reduction from (gap) $\lin$ to $\cvp_{p, \gamma}$, we get that for any $p \in [1, \infty)$, $\alpha > \alpha_p^{\ddagger}$, there is no sub-exponential time algorithm for $\bdd_{p, \alpha}$, unless the randomized Exponential Time Hypothesis is false. 

\begin{restatable}[\ETH hardness of $\bdd_{p, \alpha}$]{theorem}{thmbddethhard}
\label{thm:BddEthHard}
    For any $p \in [1, \infty)$, $\alpha > \alpha_p^{\ddagger}$, there is no $2^{o(n)}$ time algorithm for $\bdd_{p, \alpha}$ over $\R^n$, unless the randomized Exponential Time Hypothesis is false.
\end{restatable}

\textbf{$\ETH$ hardness of Minimum Distance Problem for Linear Codes.} A linear code is a subspace of $\F_q^m$, for some prime power $q$. Given a full rank generator matrix $C \in \F_q^{m \times n}$, such that $1 \leq n \leq m$, the $q$-ary code generated by $C$ is 
\[ \mathcal{C} := C \F_q^n = \set{C\vec z : z \in \F_q^n}. \]
The elements of $\mathcal{C}$ are known as {\em codewords}.
The nearest codeword problem ($\problem{NCP}$) asks to find the minimal Hamming distance between a given target vector and a codeword in a given linear code. This problem can be thought of as the code equivalent of the closest vector problem ($\cvp$) over lattices. 
Similarly, the $\svp$ equivalent over codes is the minimum distance problem for linear codes ($\problem{MDP}$), that asks to find the minimal Hamming weight of a non-zero code word in a given code. 
Observe that the $\lin$ problem is essentially the $\problem{NCP}$ problem in disguise. 

In \cite{SDV19}, the authors studied the $\problem{SETH}$ and $\geth$ hardness of $\problem{MDP}$ and $\problem{NCP}$. They give a reduction from $k$-$\problem{SAT}$ (Max-$k$-$\problem{SAT}$) to ($\gamma$-approximate) $\problem{NCP}$, establishing $\problem{SETH}$ ($\geth$) hardness of ($\gamma$ approximate) $\problem{NCP}$. Here $\gamma > 1$ is a constant. 
Moreover, they give a reduction from $\gamma$-$\problem{NCP}$ in rank $n$ to $\gamma'$-$\problem{MDP}$ in rank $Cn$ for constants $C$, $\gamma'$, thereby showing $\geth$ hardness of $\gamma'$-$\problem{MDP}$. 
We note that this reduction can be thought of as a reduction from $\lin_\varepsilon$ for a constant $\varepsilon$ in $n$ variables to $\gamma$-$\problem{MDP}$ in rank $n$, for another constant $\gamma$. Thus, we get $\ETH$ hardness of $\gamma$-$\problem{MDP}$ for some constant $\gamma$ for free.

\subsection{Our Techniques}
\label{ssec:tech_overview}

\textbf{From Linear Equations to $\cvp$.}
$\lin$ is one of the classical $\np$-hard approximation problems \cite{Has01}. An instance of $\lin$ is of the form $(M, \vec v)$, where $M \in \F^{m \times n},~\vec v \in \F^{m}$ have entries from a finite field $\F$. 
The goal is to find a vector $\vec x \in \F^n$ that satisfies as many linear equations $M_i \cdot \vec x =  v_i$  as possible. 
The gap version $\problem{gap}_{c,s}\lin$ is a promise problem where an instance $(M, \vec v)$ is a \yes instance if there exists a vector $\vec x \in \F^n$ that satisfies at least $cm$ of the linear equations, and a \no instance if every vector $\vec x \in \F^n$ satisfies at most $sm$ of the linear equations.
\cite{BHIRW24} showed that there is a Karp reduction from $\sat$ to $\problem{gap}_{c,s}\lin$ for $c = 5/8$ and $s = 5/8 - \varepsilon$ for an explicit constant $\varepsilon > 0$. This implies that $\problem{gap}_{c,s}\lin$ is \ETH hard. From now on, we write $\lin_\varepsilon$ to denote $\problem{gap}_{c,s}\lin$ over $\F_2$, for $c = 5/8$ and $s = 5/8 - \varepsilon$. 

To prove \cref{thm:CvpEthHard}, we show a deterministic Karp reduction from $\lin_\varepsilon$ to $\cvp_{p,\gamma}$ for a constant $\gamma := \paren{1 + 8\varepsilon / 3}^{1/p}$. Consider the following matrix-vector pair
\begin{equation}
    \label{peqn:cvp_instance}
    B := [M \quad 2I_m]; \quad \vec t := \vec v,
\end{equation}
where $I_m$ is the identity matrix. Consider the lattice generated by $B$. 
Appending $2I_m$ to $M$ ensures that distances between the lattice points in $\lat(B)$ to the target vector $\vec t$ are computed modulo 2. Precisely, the presence of $2I_m$ enforces that for any lattice point $B\vec{x}$, the coordinates of the difference $B\vec{x} - \vec t$ remains within an equivalence class mod 2, effectively reducing the computations to over $\F_2$. Therefore, we are able to show that if the input was a \yes instance of $\lin_\varepsilon$, then there exists a lattice vector in $\lat(B)$ that is at distance at most $r := (3m/8)^{1/p}$ from the target vector $\vec t$; and if the input was a \no instance of $\lin_\varepsilon$, then for all lattice vectors in $\lat(B)$, their distance to the target vector $\vec t$ is at least $\gamma r$. Thus, $(B, \vec t, r)$ is an instance of $\cvp_{p,\gamma}$.

\textbf{Enter Sparsification: From $\lin$ to $\svp$.}
Aggarwal and Stephens-Davidowitz \cite{AS18} showed $\geth$ hardness of $\svp_{p, \gamma}$ for all $p > 2$. Using ideas from Khot \cite{Khot05svp}, they reduced an instance of (gap) $\problem{ExactSetCover}$ over $m$ sets and a universe of size $k$, to $\svp_{p, \gamma}$, using an auxiliary {\em gadget}, a lattice-target pair $(B^{\dagger}, \vec{t}^{\dagger})$ of dimension $\dgr d$. 
Precisely, they used the input $\problem{ExactSetCover}$ instance to define a lattice $\widehat{\lat} \subset \R^{m + k}$ generated by a matrix $\widehat{B}$, and a vector $\widehat{\vec t}$. Define 
\[ B := \begin{pmatrix}
    \widehat{B} & 0 \\
    0 & \dgr B
\end{pmatrix}; \quad \vec t := \begin{pmatrix}
    \widehat{\vec t} \\
    \vec{t}^{\dagger}
\end{pmatrix}. \]

Denote the number of vectors at a radius at most $r$ around $\vec t$ in the lattice $\lat(B)$ by $\nump{\lat(B), r, \vec t}$. 
They show that for a certain choice of $\dgr B$, if the $\problem{ExactSetCover}$ instance is a \yes instance, then $\nump{\lat(B), r, \vec t}$ is exponentially larger than $\nump{\lat(B), \gamma r, \vec t}$ if it were a \no instance. 
Then they define a (generating set of a) lattice 

\[ B' := \begin{pmatrix}
    B & \vec t \\ 
    0 & s 
\end{pmatrix}, \]
for some $s$, and use a lattice {\em sparsification} algorithm (\cref{ssec:sparcification_prelims}) to sample a random (sparser) sub-lattice $\lat'' \subseteq \lat(B')$ such that if the input $\problem{ExactSetCover}$ instance was a \yes instance, at least one lattice vector of length at most $r$ survives in $\lat''$, and if it were a \no instance, then all lattice vectors of length at least $\gamma r$ die. 
Together with a reduction from $\problem{Gap3SAT}$ to $\problem{ExactSetCover}$, they get the $\geth$ hardness of $\svp_{p, \gamma}$ for all $p > 2$.

They pick the gadget $(\dgr B, \vec{t}^{\dagger})$ so that if the input $\problem{ExactSetCover}$ instance was a \yes instance, then it blows up the number of short lattice vector exponentially, whereas if it were a \no instance, the number of short lattice vectors remains small. This corresponds to a gadget that has exponentially more close vectors than short vectors. (We usually call such a lattice-target pair a {\em locally dense} gadget.) They show that for all $p > 2$, the integer lattice $\Z^d$, along with the vector $\vec t := t \cdot \vec{1}_d$ for some $t \in (0, 1/2]$ satisfies this property. To see this, for any $\tau > 0, t \in [0, 1/2]$ define the theta function

\[ \Theta_p(\tau, t) := \sum_{z \in \Z} \exp\paren{-\tau \abs{z - t}^p}. \]

Notice that 

\begin{align*}
    \Theta_p(\tau, \vec t) := \Theta_p(\tau, t)^d &= \sum_{z_1, \ldots, z_d \in \Z} \exp\paren{-\tau \sum_{i=1}^d \abs{z_i - t}^p} \\
    &\geq \sum_{\substack{\vec z \in \Z^d \\ \norm{\vec z - \vec t}_p \leq r }} \exp\paren{-\tau \norm{\vec z - \vec t}_p^p} \\
    &\geq \exp\paren{-\tau r^p} \cdot \nump{\Z^d, r, \vec t}.
\end{align*}

This implies that
\begin{equation}
    \nump{\Z^d, r, \vec t} \leq \exp\paren{\tau r^p} \cdot \Theta_p(\tau, t)^d.
    \label{eqn:zdt}
\end{equation}

Therefore, the theta function can be used to find an upper bound to the number of lattice vectors close to $\vec t$ in the $\ell_p$ norm.
In fact, the above inequality is quite strict. It has been shown \cite{MO90,EOR91,AS18} that $\Theta_p(\tau, \vec t)$ can be used to approximate the number of integer points in an $\ell_p$ ball up to sub-exponential factors. 
Thus, there exists a vector of the form $t \cdot \vec{1}_d$ for some $t \in (0, 1/2]$ such that there are exponentially more close lattice vectors in the integer lattice than short lattice vectors if and only if there exists a $\tau > 0$ and a $t \in (0, 1/2]$ such that $\Theta_p(\tau, t) >  \Theta_p(\tau, 0)$. They show that this is true for all $p > 2$ \cite[Section 6]{AS18}.

Unfortunately, a similar reduction breaks down if we start from $\lin_\varepsilon$ instead of $\problem{ExactSetCover}$. This is because while counting the number of short ({\em annoying}) vectors in the no case (say $A$), they \cite{AS18} exploit the fact that for any non-zero integer $\cl$, $\dist_p(\widehat{\lat},\cl \widehat{\vec{t}})$ is sufficiently large. Therefore they only have to count the number of vectors in $\lat(\dgr B)$ around $\cl \vec{t}^{\dagger}$ in a much {\em smaller} radius. 
When we define the corresponding $\cvp$ instance in \cref{peqn:cvp_instance}, for every even integer $2\cl$, the vector $2\cl \vec t \in \lat(B)$, and therefore we cannot use a similar trick.

We can, however, construct a gadget that has a stronger property: for every {\em even} multiple of the target, $2\ell \vec t^{\dagger}$, it holds that the number of lattice points in $\lat(\dgr B)$ around $\vec t^{\dagger}$ are exponentially more than those around $2\ell \vec{t}^{\dagger}$. This would compensate for the increase in radius incurred in our situation. 
In \cref{ssec:all_half_target} we show that for all $p > 2$ the integer lattice $(\lat(\dgr B) = \Z^d)$ and the target $\vec{t}^\dagger = \frac{1}{2}\vec{1}_d$ satisfies this property. 
This is because with this choice, every even multiple would be an integer vector. Thus, for some $r'<r$, if the gadget satisfies the property that $\nump{\lat(\dgr B), r', \vec{t}^{\dagger}}$ is exponentially more than $\nump{\lat(\dgr B), r, \vec 0}$, then for every even integer $2\ell$, it holds that $\nump{\lat(\dgr B), r', \vec{t}^{\dagger}}$ is exponentially more than $\nump{\lat(\dgr B), r, 2\ell \vec{t}^{\dagger}}$. 

One of the main technical ingredients, which could be of independent interest, is \cref{thm:thetaProp}, where we use Fourier analysis to show that for any $p > 2$, there always exists $\tau > 0$ such that
$$
\Theta_p(\tau, 1/2) > \Theta_p(\tau, 0).
$$
Since $\Theta_p(\tau, 1/2)$ and $\Theta_p(\tau, 0)$ are absolutely convergent, this can be proved by analyzing the difference series \[f_p(\tau) := \Theta_p(\tau, 1/2) - \Theta_p(\tau, 0)=-1+2\sum_{z=1}^{\infty}(e^{-\tau(z-1/2)^{p}}-e^{-\tau z^{p}}).\] In other words, it suffices to show that for all $p>2$, there always exists $\tau>0$ such that $f_{p}(\tau)>0$. When $p$ is significantly bounded away from $2$, e.g., for $p\geq 2.2$, we can simply take a constant $\tau=2$, which yields \[f_p(2)>-1+2(e^{-2(1/2)^{p}}-e^{-2})\geq -1+2(e^{-2(1/2)^{2.2}}-e^{-2})>0,\] since $e^{-2(z-1/2)^{p}}-e^{-2 z^{p}}>0$ for all $z\in\Z^+$. However, as $p$ gets closer to 2, if we still fix some constant $\tau$, we will need more than just the first term of the series, but the required number of terms will depend on $p$, making such an analysis no longer feasible.

Before illustrating our proof idea for $p$ close to $2$, let us first discuss the case for $p=2$. Interestingly, the above conclusion is known to be false for $p=2$ by \cite[Corollary 1.2]{Ban95}, i.e., $f_{2}(\tau)< 0$ for all $\tau>0$. We now introduce a novel approach to handle the $p=2$ case, which will also be instructive for the $p>2$ case. The fourth Jacobi theta function is defined by \[\vartheta_4(x, q) \coloneqq \sum_{z\in \Z}(-1)^z q^{z^2} e^{2 z i x}.\] Note that $f_2(\tau)=-\vartheta_4(0, e^{-\tau/4})$. It is well-known that the fourth Jacobi theta function $\vartheta_4(x,q)$ can be expressed in terms of the second Jacobi theta function \[\vartheta_2(x,q):=\sum_{z\in\Z} q^{(z+1 / 2)^2} e^{(2 z+1) i x}\] by Jacobi's imaginary transformation \cite{Jacobi1828}. Specifically, we can obtain
$$
\vartheta_4(0, e^{-\tau/4}) = \sqrt{\frac{4\pi}{\tau}}\vartheta_2(0,e^{-\frac{4\pi^2}{\tau}})=\sqrt{\frac{4\pi}{\tau}}\cdot 2 \sum_{z=0}^{\infty} e^{-(4\pi^2/\tau)\cdot(z+1 / 2)^2} > 0,
$$
	which implies $f_2(\tau)<0$ for all $\tau>0$ and concludes the $p=2$ case. Essentially, the trick relies on Jacobi's imaginary transformation, which turns an alternating series into a non-alternating series with all positive terms. So, this motivates us to find a transformation that can also rewrite $f_p(\tau)$ for $p>2$ in the spirit of the principles behind Jacobi's imaginary transformation. Note that Jacobi's imaginary transformation is obtained by the Poisson summation formula \cite{Poisson1827}, which states that for any function\footnote{$L^{1}(\mathbb{R})$ denotes the space of Lebesgue integrable functions on $\mathbb{R}$, i.e., $f\in L^{1}(\mathbb{R})$ if $\int_{-\infty}^{\infty} |f(x)| \, dx < \infty$.} $s(z)\in L^1(\R)$, if its Fourier transform $\hat{s}(x)\in L^1(\R)$, then 
	$$
	\sum_{z\in\mathbb{Z}}s(z)=\sum_{x\in\mathbb{Z}}\hat{s}(x).
	$$
Following this new perspective, we outline the proof for the $p>2$ case as follows. We first slightly rewrite the difference series $f_p(\tau) := \Theta_p(\tau, 1/2) - \Theta_p(\tau, 0)$ to match the Poisson summation formula and let $p = 2 + \varepsilon$ for $\varepsilon > 0$:
\begin{align*}
    f_{2 + \varepsilon}(\tau) &= -1+2\sum_{z=1}^{\infty}(e^{-\tau(z-1/2)^{p}}-e^{-\tau z^{p}}) \\
    &= - \sum_{z \in \Z} \underbrace{e^{{(-\tau / 2^{2 + \varepsilon}) |z|^{2 + \varepsilon} + \pi i z}}}_{\eqqcolon s_{\tau, \varepsilon} (z)}.
\end{align*}
By \cref{lem:1,lem:2}, we show that both $s_{\tau, \varepsilon}(z),\hat{s}_{\tau, \varepsilon}(x)\in L^1(\R)$, to ensure that the Poisson summation formula $\sum_{z\in\Z}s_{\tau, \varepsilon}(z)=\sum_{x\in\Z}\hat{s}_{\tau, \varepsilon}(x)$ can be applied. Then it suffices to show that for any $\varepsilon > 0$, there exists some sufficiently small $\tau > 0$ such that for every $x\in\Z$, $\hat{s}_{\tau, \varepsilon}(x) < 0$. To find such $\tau$, we consider the asymptotic expansion for the Fourier transform of the super-Gaussian $g(z):=e^{-\abs{z}^{2+\varepsilon}}$, i.e., $\hat{g}_{\varepsilon}(x) := \int_{-\infty}^{\infty} e^{-\abs{z}^{2+\varepsilon} - 2 \pi ix z} d z$. This has been studied in \cite{EOR91} (restated as \cref{lem:eor}), which demonstrates that $\hat{g}_{\varepsilon}(x)<0$ for all sufficiently large $x$. We then relate $\hat s_{\tau, \varepsilon}(x)$ to $\hat{g}_{\varepsilon}(x')$ where $x'$ is obtained from $x$ by a translation and a scaling that depend on $\tau$ and $\varepsilon$. Thus, for every $\varepsilon>0$, we can choose an appropriate $\tau$ such that $x'$ is large enough for all $x\in\Z$, and thereby conclude that $\hat s_{\tau, \varepsilon}(x)<0$ for all $x\in\Z$. This implies that for every $\varepsilon>0$, the difference series $f_{2+\varepsilon}(\tau)=-\sum_{x\in\Z}\hat{s}_{\tau, \varepsilon}(x)>0$ for some $\tau>0$, which completes the proof for \cref{thm:thetaProp}. 

Putting them together, we get \cref{thm:SvpIntegerGadget}, which says that for any $p > 2$, we can find a radius $r$ such that for any $d$, the integer lattice contains exponentially many points within a distance $r$ around $1/2 \cdot \vec{1}_d$ than around any of its {\em even multiple}.
Using these gadgets in our reduction, \cref{thm:SvpEthHard} follows.  

We leave $p = 2$ as an open problem. Note that the $\ETH$ hardness of $\svp_{p, \gamma}$ for $p \in [1, 2)$ follows from the $\ETH$ hardness of $p = 2$ by the norm embedding techniques of \cite{RR06}.

\textbf{Open Problem 1:} Prove that for any $p \in [1,2]$, there exists a constant $\gamma>1 $ such that assuming $\ETH$ there are no sub-exponential time algorithms for $\svp_{p, \gamma}$.
\vspace{10pt}

\textbf{On to $\bdd$.}
Bounded distance decoding is a lattice problem that has found applications in showing hardness results for important cryptographic primitives such as {\em Learning With Errors} $\problem{(LWE)}$. 
It can be thought of as $\cvp$ under a promise that the closest lattice vector to a target $\vec t$ is not too far away from the lattice $\lat(B)$, relative to the length of the shortest non-zero vector. Alternatively, it can be thought of as a {\em decoding} problem over lattices, analogous to decoding noisy code words over finite fields. Quantitatively, a $\bdd_{p, \alpha}$ instance promises that there is a closest lattice vector at a distance at most $\alpha \lambda_1^{(p)}(\lat(B))$ from the lattice. (See \cref{def:bdd} for the formal definition.)
Regev \cite{oded05}, in a seminal work, gave a reduction from worst case $\bdd$ to (an average case problem) $\problem{LWE}$, with polynomial (in the dimension of the lattice) $\alpha$. 
It is easy to see that there would be a unique solution to the problem if $\alpha < 1/2$.  
If $\alpha_1 > \alpha_2 > 0$, it is easy to see that $\bdd_{p, \alpha_2}$ reduces to $\bdd_{p, \alpha_1}$. Therefore, when showing hardness results for $\bdd$, a lower $\alpha$ corresponds to a stronger result. 
$\np$-hardness of $\bdd$ for a constant $\alpha$ was first shown by \cite{LLM06}, with $\alpha = \min \{2^{-1/2}, 2^{-1/p}\}$ for $p \geq 1$. This reduction however incurs a polynomial blowup in the rank of the lattice.
Recently \cite{BDD20} studied the quantitative hardness of $\bdd$ for the first time and show that for all $p > 1$, there is no $2^{\Omega(n)}$ time algorithm for $\bdd_{p, \alpha}$ for all $\alpha$ greater than a certain constant that approaches $1/2$ as $p \rightarrow \infty$, unless randomized \ETH is false.
In a followup work \cite{bdd21}, they show a similar result for improved values of $\alpha$, under the $\geth$ assumption. Quantitatively, they define a constant $\alpha_p^{\ddagger}$ \cref{eqn:alphapddagger} that depends on $p$, and show that for all $p > 1$, for all $\alpha > \alpha_p^{\ddagger}$ there is no sub-exponential time algorithm for $\bdd_{p, \alpha}$ unless $\geth$ is false. 
To get this result, they show a decision-to-search reduction from $\cvp'_{p, \gamma}$ to the decision version of $\bdd_{p, \alpha}$, where the goal is to decide if there is a vector at a distance at most $\alpha \cdot \lambda_1^{(p)}(\lat(B))$. Here, $\cvp'$ is a special case of $\cvp$ where in the $\yes$ case, there is a {\em binary} vector $\vec x \in \set{0, 1}^n$ such that $\norm{B\vec x - \vec t}_p \leq r$, as compared to an arbitrary integer vector $\vec x$.
For this they make use a $\cvp'_{p, \gamma}$ instance constructed from $\sat$ in \cite{BGS17}. Precisely, given a rank $n'$ instance of $\cvp'$ $(B', \vec t', r')$, for some scalar parameters $s, l > 0$, they define a (generating set of a) lattice and target pair 

\[ B := \begin{pmatrix}
    sB' & 0 \\
    I_{n'} & 0 \\
    0 & l \dgr B 
\end{pmatrix}; 
\quad \quad \vec t := \begin{pmatrix}
    s \vec t' \\
    \frac{1}{2} I_{n'} \\
    l \vec{t}^{\dagger}
\end{pmatrix}, \]
where $(\dgr B, \dgr{\vec t})$ is again a locally dense lattice gadget. Similar to what we saw in case of $\svp$ above, they are able to bound the number of lattice vectors close to the target $\vec t$ in the \yes case by, say $G$, which is exponentially larger than the number of lattice vectors close to $\vec t$ in the \no case. 
However, there is a technical difficulty here. We also need the fact that if the output $\bdd$ instance $(B, \vec t, r)$ is a \yes instance, it satisfies the promise that the nearest lattice vector is at a distance at most $\alpha \cdot \lambda_1^{(p)}(\lat(B))$. 
This corresponds to also bounding the number of very short vectors in the gadget in the \yes case by some $A$. As before, they can then sparsify the lattice (\cref{ssec:sparcification_prelims}), to get a sparse random sub-lattice such that if the input $\cvp'$ instance was a yes instance then there is a lattice vector at a distance at most $\alpha \cdot \lambda_1^{(p)}(\lat(B))$ from the target, with high probability.

In this work we are able to achieve the same lower bound of $\alpha_p^{\ddagger}$ from \cite{bdd21}, {\em but under $\ETH$}, by starting from a $\cvp_{p, \gamma}$ instance constructed from gap $\lin$. To achieve this, we show a reduction from arbitrary $\cvp_{p, \gamma}$ instances over lattices over integers, $\lat \subseteq \Z^n$. In fact, we are able to get a {\em simpler} reduction, in that we don't need to embed an integer lattice into our $\bdd$ instance anymore. Given a $\cvp_{p, \gamma}$ instance $(B', \vec t', r)$ we can define the following lattice 

\[ B := \begin{pmatrix}
    B' & 0 \\
    0 & s \dgr B 
\end{pmatrix}; 
\quad \quad \vec t := \begin{pmatrix}
     \vec t' \\
    s \vec{t}^{\dagger}
\end{pmatrix}, \]

where $(\dgr B, \dgr{\vec t})$ is again a locally dense lattice gadget. 
Using an analysis similar to that in case of $\svp_{p, \gamma}$ in \cref{sec:svp}, and a locally dense integer gadget from \cite{bdd21}, we get a desired reduction from $\cvp_{p, \gamma}$ to $\bdd_{p, \alpha}$, for all $p \geq 1$ and for all $\alpha > \alpha_p^{\ddagger}$. 
When we instantiate this reduction with the $\cvp_{p, \gamma}$ instance from \cref{sec:cvp}, we find that there is no sub-exponential time algorithm for $\bdd_{p, \alpha}$, for all $p \geq 1$ and for all $\alpha > \alpha_p^{\ddagger}$, unless the (randomized) $\ETH$ is false.

	\section{Preliminaries}
	\label{sec:prelims}
	Whenever we say that certain constants are {\em efficiently computable}, we mean that they can be efficiently {\em approximated} to high precision. 

For any set $\cL$ we define $\ballp{\cL, r, \vec t} := \mathcal{B}_p(r, \vec t) \cap \cL$, where $\mathcal{B}_p(r, \vec t)$ denotes a ball in $\R^n$ centered at $\vec{t}$ of radius $r$ in the $\ell_p$ norm, i.e., $\ballp{r, \vec t} := \set{\vec x \in \R^n : \norm{\vec x - \vec t}_p \leq r}$.
For a discrete set $\cL$, we define \[
\nump{\cL, r, \vec{t}} := \abs{\mathcal{B}_p(r, \vec t) \cap \cL}.
\]
For any matrix $B \in \R^{m \times n}$, we write $\lat(B)$ to denote the lattice generated by the columns of $B$. 
We write $\vec 1_n$ and $\vec 0_n$ to denote the vector all $1$s and all $0$s vectors in $n$ dimensions respectively.
We write $\lambda_1^{(p)}(\lat)$ for the length of the {\em shortest non-zero vector with respect to $\cl_p$-norm} in the lattice $\lat$. 
For vectors $\vec v_1 \in \R^n,~\vec v_2 \in \R^m$, we write $(\vec v_1, \vec v_2)$ for their concatenation: $(\vec v_1^\top, \vec v_2^\top)^\top$. Unless otherwise specified, all logarithms are base $e$. 
For a discrete set $\lat \subset \R^n$ and a vector $\vec t \in \R^n$ we define the distance between them to be the minimum distance between the vector $\vec t$ and any point in the set:
\[ \dist_p(\vec t, \lat) := \min_{\vec x \in \lat} \norm{\vec x - \vec t}_p. \]

\subsection{Asymptotic Expansion}

We now define the notion of Poincar\'e Asymptotic Expansion, which we will use in \cref{ssec:all_half_target}. An asymptotic expansion describes the asymptotic behavior of a function in terms of a sequence of (gauge) functions \cite[Definition 2.3]{Hunter2004Asymptotic}.

\begin{definition}[Poincar\'e Asymptotic Expansion]\label{def:asymptotic}
    A sequence of functions $\varphi_n : \mathbb{R} \setminus 0 \to \mathbb{R}$, where $n = 0, 1, 2, \dots$, is a \textit{asymptotic sequence} as $x \to \infty$ if for each $n = 0, 1, 2, \dots$ we have
\[
\varphi_{n+1} = o(\varphi_n) \quad \text{as } x \to \infty.
\]
Moreover, if $\{\varphi_n\}$ is an asymptotic sequence and $f : \mathbb{R} \setminus 0 \to \mathbb{R}$ is a function, we write
\begin{equation}
\label{eqn:asym_exp}
f(x) \sim \sum_{n=0}^{\infty} a_n \varphi_n(x) \quad \text{as } x \to \infty
\end{equation}
if for each $N = 0, 1, 2, \dots$ we have
\[
f(x) - \sum_{n=0}^{N} a_n \varphi_n(x) = o(\varphi_N) \quad \text{as } x \to \infty.
\]

We call \cref{eqn:asym_exp} the asymptotic expansion of $f$ with respect to $\{\varphi_n\}$ as $x \to \infty$.
\end{definition}

For instance, the functions $\varphi_n(x) = x^{-n}$ form an asymptotic sequence as $x \to \infty$.
We will use the asymptotic expansion of the Fourier transform of super-Gaussian functions as shown in \cite{EOR91}. 

\begin{lemma}[\cite{EOR91}, Lemma 10 (i)]
\label{lem:eor}
		For fixed $\varepsilon>0$, the asymptotic expansion of
		$$
		\hat{g}_{\varepsilon}(x) := \int_{-\infty}^{\infty} e^{-|z|^{2+\varepsilon}-2\pi i xz} dz
		$$
		as $|x|\rightarrow\infty$ is
		$$
		\hat{g}_{\varepsilon}(x) \sim 2 \sum_{m=1}^{\infty} \frac{(-1)^{m+1}}{m!} \sin \left(\frac{m \pi (2+\varepsilon)}{2}\right) \Gamma(m (2+\varepsilon)+1)(2 \pi|x|)^{-m (2+\varepsilon)-1}.
		$$
\end{lemma}

\subsection{Computational Problems}

	For an integer $k \geq 2$, a $k$-$\SAT$ formula over $n$ boolean variables is the conjunction of clauses, where each clause is the disjunction of $k$ literals. That is, $k$-$\SAT$ formulas have the form $\bigwedge_{i=1}^m \bigvee_{j=1}^k b_{i,j}$, where $b_{i,j} = x_k$ or $b_{i,j} = \neg x_k$ for some boolean variable $x_k$. 
	
	\begin{definition} 
		For any $k \geq 2$, the decision problem $k$-$\SAT$ is defined as follows. The input is a $k$-$\SAT$ formula. It is a \yes instance if there exists an assignment to the variables that makes the formula evaluate to true and a \no instance otherwise.
	\end{definition}

	\begin{definition}
		For any $k \geq 2$, the decision problem Max-$k$-\SAT is defined as follows. The input is a $k$-\SAT formula and an integer $S \geq 1$. It is a \yes instance if there exists an assignment to the variables such that at least $S$ of the clauses evaluate to true and a \no instance otherwise.
	\end{definition}
	
	Notice that $k$-\SAT is a special case of Max-$k$-$\SAT$.
    We write $\sat_C$ for a $\sat$ instance such that it contains at most $Cn$ clauses.

    \begin{definition}\label{def:gapmaxlin}
        An instance of a $\problem{gap}_{(c,s) } \LIN$ over $\F_2$ consists of  an $m \times n$ matrix $A \in \{0,1\}^{m \times n}$ and a vector $\vec{b} = \{0,1\}^m$ constraints. It is a \yes instance if there exists an $\vec{x} = (x_1, \ldots, x_n) \in \{0,1\}^n$ that satisfies at least $c \cdot m$ of the $m$ constraints
        \[
        \sum_{j=1}^n A_{i,j}\cdot x_{j} =b_i \pmod 2 \;, \] 
        for $i = 1$ to $m$.
        It is a \no instance if for all $\vec{x } \in \Z^n$,\footnote{Note that this is equivalent to making the same statement for all $\vec{x} \in \{0,1\}^n$ since the equations are considered modulo $2$.} at most $s \cdot m$ constraints are satisfied. 
    \end{definition}

	\begin{definition}[Shortest Vector Problem ($\SVP_{p, \gamma}$)]
		\label{def:svp}
		For any $p \in [1, \infty]$ and any $\gamma \geq 1$, \emph{the $\gamma$-approximate Shortest Vector Problem in the $\ell_p$ norm} $(\SVP_{p, \gamma})$ is a promise problem defined as follows. The input is a matrix $B \in \Q^{d\times n}$ generating a lattice $\lat \subset \R^d$ of rank $n$ and a length $r > 0$. It is a \yes instance if $\lambda_1^{(p)}(\lat) \leq r$ and a \no instance if $\lambda_1^{(p)}(\lat) > \gamma r$.
	\end{definition}
	
	\begin{definition}[Closest Vector Problem ($\cvp_{p, \gamma}$)]
		\label{def:cvp}
		For any $p \in [1, \infty]$ and any $\gamma  \geq 1$, \emph{the $\gamma$-approximate Closest Vector Problem in the $\ell_p$ norm} $(\cvp_{p, \gamma})$ is a promise problem defined as follows. The input is a matrix $B \in \Q^{d\times n}$ generating a lattice $\lat \subset \R^d$ of rank $n$, a target $\vec{t} \in\R^d$, and a distance $r > 0$. It is a \yes instance if $\dist_p(\vec{t}, \lat) \leq r$ and a \no instance if $\dist_p(\vec{t}, \lat) > \gamma r$.
	\end{definition}

	Note that in \cref{def:svp} and \cref{def:cvp}, the input is a matrix $B$ that {\em generates} the lattice $\lat$. This is equivalent to the more standard definition where the input is a basis, i.e. a set of linearly independent vectors that generates  the lattice. Given a generating set $B$, a basis can be efficiently (in the bit length of the representation of $B$) computed from the generating set using the LLL algorithm \cite{LLL82} (c.f. \cite[Algorithm 1]{rotateZ23}.)

    \begin{definition}[Bounded Distance Decoding ($\bdd_{p, \alpha}$)]
    \label{def:bdd}
For $p \in [1, \infty]$ and $\alpha = \alpha(n) > 0$, the search problem $\bdd_{p,\alpha}$ is defined as follows. The input is a (generating set for a) lattice $\cL \subset \R^d$ and a target $\vec t \in \mathbb{R}^d$ satisfying
\[
\dist_p(\vec t, \lat) \leq \alpha \cdot \lambda_1^{(p)}(\cL),
\]
and the goal is to find a closest lattice vector $\vec v \in \cL$ to $\vec t$ such that
\[
\|\vec t - \vec v\|_p = \dist_p(\vec t, \lat).
\]
\end{definition}
    

	\subsection{Fine-grained Complexity}
	\label{sec:fine-grained_prelims}
	
    \begin{theorem}[\cite{BHIRW24}, Theorem 6.3]
     \label{thm:satToLinRed}
        For some $\varepsilon \in (0,1)$ there exists a polynomial time reduction from $\sat_C$ with $n$ variables to $\mathsf{gap}_{c, s}\lin$ over $\F_2$ with $\order{n}$ variables and $\order{n}$ equations, where $c = 5/8$ and $s = 5/8 - \varepsilon$.
    \end{theorem}

    Throughout this text, we write $\lin_\varepsilon$ for $\mathsf{gap}_{c, s}\lin$ over $\F_2$ with $c = 5/8$ and $s = 5/8 - \varepsilon$.
	Impagliazzo and Paturi introduced the following celebrated and well-studied hypothesis concerning the fine-grained complexity of $k$-$\SAT$~\cite{IP01}. We will also need the randomized variant, which talks about the existence of randomized algorithms instead of deterministic ones. 
	\begin{definition}[Exponential Time Hypothesis ($\ETH$)]
	    The (randomized) Exponential Time Hypothesis ((randomized) ETH) asserts that, there exists $\delta >0$ such that any (randomized) algorithm which solves $3$-\SAT must take $2^{\delta n}$ time.
	\end{definition}
\begin{definition}[$\geth$] The (randomized) \geth asserts that there exists $\delta>0$ and $0<\eta<1$ such that given a $3$-$\SAT$ instance with $n$ variables and $m$ clauses, any (randomized) algorithm which can distinguish between the cases if all $m$ clauses are satisfiable and one in which no assignment satisfies more than $\eta$-fraction of the clauses, must take $2^{\delta n}$ time.
\end{definition}
    \begin{lemma}[Sparsification Lemma \cite{IP01}]\label{lemma:sp}
    Let $\varepsilon > 0$, $k \geq 3$ be constants. There is a $2^{\varepsilon n} \cdot \text{poly}(n)$ time algorithm that takes a $k$-CNF $F$ on $n$ variables and produces $F_1, \dots, F_{2^{\varepsilon n}}$, $2^{\varepsilon n}$ $k$-CNFs such that $F$ is satisfied if and only if $\bigvee_i F_i$ is satisfied and each $F_i$ has $n$ variables and $n \cdot \left(\frac{k}{\varepsilon}\right)^{O(k)}$ clauses. In fact, each variable is in at most $\text{poly}\left(\frac{1}{\varepsilon}\right)$ clauses, and the $F_i$ are over the same variables as $F$.
    \end{lemma}

The sparsification \cref{lemma:sp} and Tovey's reduction \cite{TOVEY} together tell us that if $\ETH$ is true, then $3SAT_4$ over $n$ variables can't be solved in  $2^{o(n)}$ time.
Together with \Cref{thm:satToLinRed}, we find that if $\ETH$ holds, then for some $\varepsilon>0$, any algorithm which solves $\lin_\varepsilon$ with $n$ variables, $m=O(n)$ clauses, must take $2^{\delta n}$ time for some $\delta > 0$. We state it as the following corollary.

\begin{corollary}[$\lin_{\varepsilon}$ is $\ETH$ Hard]
	\label{cor:linEthHard}
	There exists constants $\varepsilon > 0, C > 0$ such that unless $\mathsf{ETH}$ is false, there is no $2^{o(n)}$-time algorithm for $\lin_\varepsilon$ with $n$ variables and $m= Cn$ equations.
\end{corollary}

\subsection{Counting Lattice Points}

We now define the $\Theta$ and the $\mu$ functions, and show that they can be used to approximate the number of lattice points within a given radius.

For any $p \in [1, \infty), \tau>0$, and  $t \in \R$ define the theta function to be
\[\Theta_p(\tau, t):= \sum_{z \in \Z} \exp(-\tau |z-t|^p).\]

Notice that without loss of generality, we can assume $t \in [0,1/2]$. For a vector $\vec{t} \in \R^n$ we can analogously define 
$$\Theta_p(\tau, \vec{t}):= \prod_{i \in [n]} \Theta_p(\tau,t_i)$$
Clearly the theta function\footnote{Notice that this is closely related to the discrete Gaussian function $\rho_s(\Z - \vec t) := \sum_{\vec z \in \Z^n} \exp\paren{-\pi \norm{\vec z - \vec t}^2 / s^2}$.}
\[ \Theta_p(\tau, \vec{t}) =  \sum_{\vec{v} \in \Z^n} \exp(-\tau \norm{\vec{v}-\vec{t}}_p^p)\]
The Theta function acts as a smooth proxy for the point counting function, where the parameter $\tau$ plays the role of inverse radius.
   
For any $p \in [1, \infty), \tau>0$ and $t \in [0,1/2]$, define 
    $$ \mu_p(\tau , t) := \mathbb{E}_{X \sim D_p(\tau,t)}[|X|^p]= \frac{1}{\Theta_p(\tau,t)} \cdot \sum_{z \in \Z}|z-t|^p \cdot \exp(-\tau|z-t|^p)$$
    where, $D_p(\tau,t)$ is the probability distribution over $\Z -t$ that assigns probability $\exp(-\tau |x|^p)/\Theta_p(\tau,t)$ to $x \in \Z-t$. 
    For $\vec{t} \in \R^n$ this  extends as following
    $$\mu_p(\tau,\vec{t}):= \sum_{i=1}^n \underset{X \sim D_p(\tau,t_i)}{\expect} [|X|^p] $$

    Notice that if $\vec t = t \cdot \vec 1_n$ for some $t \in [0, 1/2]$, we have that $\mu(\vec t) = n \mu( t)$.
   
\begin{lemma}\label{lemma:theta0}
    For any $p \in [1, \infty)$, $r>0$ ,$\tau>0$ and $\vec{t}\in\R^n$, we have that $$N_p(\Z^n,r,\vec{t})\leq \exp(\tau r^p)\Theta_p(\tau,\vec{t})  $$
\end{lemma}
\begin{proof}
    \begin{align*}
       \Theta_p(\tau, \vec{t})& =  \sum_{\vec{v} \in \Z^n} \exp(-\tau \norm{\vec{v}-\vec{t}}_p^p) \\
        &\geq \sum_{\vec{v} \in \ballp{\Z^n, r, \vec{t}}}  \exp(-\tau \norm{\vec{v}-\vec{t}}_p^p)\\
       &\geq N_p(\Z^n, r , \vec t) \cdot \exp(-\tau r^p)
    \end{align*}
\end{proof}

The upper bound in the previous lemma is quite tight. In fact, we know the following theorem. 

    \begin{theorem}[\cite{AS18}, Theorem 6.1]
    \label{thm:theta_gives_good_approx}
		For any constants $p \geq 1$ and $\tau > 0$, there is another constant $C^* > 0$ such that for any $\vec{t} \in \R^n$ and any positive integer $n$.
		\[
		\exp\left(\tau \mu_p(\tau, \vec{t}) -C^*\sqrt{n} \right) \cdot  \Theta_p(\tau, \vec{t}) \;\leq \; N_p \left(\Z^n, \mu_p(\tau, \vec{t})^{1/p}, \vec{t}\right) \; \leq \; \exp\left(\tau \mu_p(\tau, \vec{t}) \right) \cdot  \Theta_p(\tau, \vec{t})
		\]
    \end{theorem}

Next we define the $\beta$ function, which we will use to define the threshold $\alpha_p^{\ddagger}$ above which we are able to show hardness results for $\bdd_{p, \alpha}$.

\begin{definition}[\cite{bdd21}] 
For $p \in [1, \infty)$, $t \in [0,1/2]$, and $a \geq 0$, we define $\beta_{p,t}(a)$ as follows.
\begin{enumerate}
    \item For $a < t$, define $\beta_{p,t}(a) := 0$.
    \item For $a = t$, define $\beta_{p,1/2}(1/2) := 2$ and for $t \neq 1/2$ define $\beta_{p,t}(t) := 1$.
    \item For $a > t$, define
    \[
    \beta_{p,t}(a) := \exp(\tau^* a^p) \cdot \Theta_p(\tau^*,t),
    \]
    where $\tau^* > 0$ is the unique solution to $\mu_p(\tau^*,t) = a^p$.
\end{enumerate}
\end{definition}

\begin{definition}[\cite{bdd21}]
    For $p \in [1, \infty)$, define
    \begin{equation}
        \label{eqn:alphapddagger}
        \alpha_p^{\ddagger} := \inf_{\substack{t \in [0,1/2] \\ a \geq t}} \frac{a}{\beta_{p,0}^{-1}(\beta_{p,t}(a))} \;.
    \end{equation}
\end{definition}

We note that the functions $\Theta,~\beta$ and $\mu$ can be efficiently approximated to within high precision. Throughout this text, we deal with constants $A, G$, which are the number of vectors in the lattice of a particular length, and will functions of $\Theta$. Thus they will be computable efficiently to a high precision. 


\subsection{Lattice Sparsification}
\label{ssec:sparcification_prelims}

Khot introduced the idea of lattice sparsification \cite{Khot05svp}, which is a randomized process that given a lattice $\lat$ lets us sample a sub-lattice $\lat' \subseteq \lat$ that has a lot fewer points in any fixed radius, large enough $\ell_p$-ball. It works by taking a random hyperplane over a finite field $\F_q$, for some prime power $q$, and restricting the coefficients of the lattice vectors to belong to the particular hyperplane. Formally, we prove and use the following statement. 



\newcommand{\cals}{\mathcal{S}}
\begin{lemma}
	\label{lem:svpSparcification2}
	For any $p \in [1, \infty)$, there is an efficient algorithm that takes as input a (basis for a) full rank lattice $\lat \subset \R^n$ of rank $n$ and a prime number $q$, and outputs a (basis for a) sub-lattice $\lat' \subseteq \lat$ of rank $n$ such that for any finite set $\cals \subset \lat$ if $\forall \vec v_1,\vec v_2 \in \mathcal{S}$ it holds that $B^{-1}\vec v_1 \mod q$ and $B^{-1}\vec v_2 \mod q$ are non-zero and pairwise linearly independent over $\F_q$, then
	\begin{equation*}
		1 - \frac{q}{\abs{\cals}} \leq \Pr \left[ \exists \vec v \in \cals : \vec v \in \lat' \right] \leq \frac{|\cals|}{q}.
	\end{equation*}
\end{lemma}
\begin{proof}
	The algorithm samples $\vec x \in \F_q^n$ uniformly at random. It then computes and outputs a basis for the lattice defined as
	\begin{equation*}
		\lat' := \{ \vec v \in \lat : \inner{B^+ \vec v, \vec x} \equiv 0 \pmod{q} \}.
	\end{equation*}
	A basis for $\lat'$ can be computed efficiently using the algorithm from \cite{DGStoSVP}, claim 2.15. The algorithm outputs a basis for $\lat'$, hence the efficiency is clear. We will prove correctness.
	Define $N := |\cals|$ and let $\vec v_1, \ldots \vec v_N \in \mathcal{S}$.
	For each $i$, define $\vec a_i$ := $B^+ \vec v_i$. 
    Let $X_i$ be the indicator random variable for $\vec v_i \in \lat'$. Let $X = \sum_i X_i$. We have 
	\begin{equation*}
		\expect[X_i] = \Pr_{\vec{x} \sim \F_q^n} \left[ \inner{\vec a_i, \vec x} \equiv 0 \pmod{q} \right] = 1/q,
	\end{equation*}
    and $\expect[X] = N/q$. Also, $\variance[X_i] = \frac{1}{q}(1-1/q)$. Since $\vec a_i$'s are pairwise linearly independent by assumption, $\variance[X] = \frac{N}{q}(1-1/q)$. By Chebyshev's inequality, 
    \begin{equation*}
        \Pr \left[ \exists \vec v \in \cals : \vec v \in \lat' \right] = \prob{X > 0} \geq 1 - \frac{\variance[X]}{\expect[X]^2} \geq 1 - \frac{q}{N}.
    \end{equation*}
	For the upper bound, by union bound
	\begin{align*}
		\Pr \left[ \exists i \in [N] : \vec v_i \in \lat' \right]
		&= \Pr \left[ \exists i \in [N] : \inner{\vec a_i, \vec x} \equiv 0 \pmod{q} \right] \\
		&\leq \sum_{i=1}^N \Pr \left[ \inner{\vec a_i, \vec x} \equiv 0 \pmod{q} \right] \\
		&= N/q.
	\end{align*}
\end{proof}

We also use the following treatment of lattice sparsification from \cite{bdd21}. 

\begin{lemma}[\cite{bdd21}, Proposition 2.5]
	\label{lem:bdd_sparcification}
    Let $p \in [1, \infty)$, let $\mathcal{L}$ be a lattice of rank $n$ with basis $B$, let $\vec t \in \operatorname{span}(\mathcal{L})$, let $q$ be a prime, and let $r \geq 0$. Let $\vec x, \vec z \sim \mathbb{F}_q^n$ be sampled uniformly at random, and define
\[
\mathcal{L}' := \{ \vec v \in \mathcal{L} :  \inner{B^+ \vec v, \vec x} \equiv 0 \pmod{q} \}, \quad \vec t' := \vec t - B\vec z.
\]

\begin{enumerate}
    \item If $r \leq q \lambda_1^{(p)}(\mathcal{L})$, then
    \begin{equation}
        \Pr[\lambda_1(\mathcal{L}') \leq r] \leq \frac{\nump{\mathcal{L} , r, \vec 0}}{q}.
    \end{equation}
	\label{lem:sparcification:1}

    \item If $r < q \lambda_1^{(p)}(\mathcal{L}) / 2$, then ,
      
    \begin{equation}
        \Pr[\operatorname{dist}_p(\vec t', \mathcal{L}') > r] \leq \frac{q}{\nump{\mathcal{L}, r, \vec t}} + \frac{1}{q^n}.
    \end{equation}

    \item
    \begin{equation}
        \Pr[\operatorname{dist}_p(\vec t', \mathcal{L}') \leq r] \leq \frac{\nump{\mathcal{L}, r, \vec t}}{q} + \frac{1}{q^n}.
    \end{equation}
	\label{lem:sparcification:3}
\end{enumerate}
\end{lemma}


\section{\texorpdfstring{$\ETH$}{ETH} hardness of \texorpdfstring{$\cvp_{p,\gamma}$}{CVP p, gamma}}
\label{sec:cvp}
\label{sec:cvp_hard}
In the following, we show a reduction from the gap $\mathsf{MAXLIN}$ problem to the $\CVP_{p, \gamma_p}$ and thereby conclude that $\CVP_{p, \gamma_p}$ is hard under $\mathsf{ETH}$.
\begin{theorem}
\label{thm:LinCvpReduction}
For any $p \in [1, \infty)$, there exists a constant $\gamma_p > 1$ such that for all $n \in \Z^+$, there is a polynomial time Karp reduction from $\lin_\varepsilon$ in $n$ variables and $m$ equations to $\CVP_{p, \gamma_p}$ over $\R^m$.
\end{theorem}

\begin{proof}
    Fix any $p$.   Define $\gamma = \gamma_p := \paren{1 + \frac{8\varepsilon}{3}}^{1/p}$.
    
    On input a $\gmlin$ instance, $M \in \{0,1\}^{m \times n},\vec v \in \{0,1\}^m$, the reduction computes 
    \begin{equation*}
        B := \begin{bmatrix}
            M & 2I_m
        \end{bmatrix}; \quad 
        \vec t := \begin{bmatrix}
            \vec v
        \end{bmatrix},
    \end{equation*}
   and sets $r := (3m/8)^{1/p}$.
    It then outputs $(B, \vec t, r)$.
    The reduction is clearly efficient. We now show that the reduction is correct. 

   We claim that $(B, \vec t, r)$ is a \yes instance of $\CVP_{p, \gamma}$ if the input $\lin_{\varepsilon}$ instance was a \yes instance, and a \no instance of $\CVP_{p, \gamma}$ if the input $\lin_{\varepsilon}$ instance was a \no instance.
    Note that for any $\vec x \in \Z^n, \vec y \in \Z^m$, we have \[\normpp{B (\vec x, \vec y) - \vec{t}} = \normpp{M\vec x + 2\vec y - \vec{v}}\;.\]
    Suppose that the input $\lin_{\varepsilon}$ instance was a \yes instance. 
    This implies that $\exists \vec x \in \set{0,1}^n$ such that $M\vec x - \vec v \pmod 2$ has at most $3m/8$ non-zero coordinates.
    
    Also, $\exists \vec y \in \Z^m$ such that $\forall i \in [m] : (M\vec x - \vec v + 2\vec y)_i$ is $1$ if $(M\vec x - \vec v)_i$ is $1$ modulo $2$, and $0$ otherwise. Therefore, \[ \normpp{M\vec x - \vec v +2\vec y} \leq 3m/8 = r^p, \]
    and hence $(B, \vec t, r)$ is a \yes instance of $\cvp_{p, \gamma}$.
    
    Next, suppose that the input $\lin_{\varepsilon}$ instance was a \no instance.

    Then we have that $\forall \vec x \in \Z^n$, $M\vec x - \vec v$ has at least $(\gamma r)^p = (3/8 + \eps)m$ co-ordinates that are non-zero modulo $2$ (odd coordinates). 
    Then, $\forall \vec y \in \Z^m$, $M\vec x - \vec v +2\vec y$ has at least $(\gamma r)^p$ non-zero integral coordinates. Therefore, \[ \normpp{M\vec x - \vec v +2\vec y} \geq (\gamma r)^p, \]
    and hence $(B, \vec t, r)$ is a \no instance of $\cvp_{p, \gamma}$.
\end{proof}

Together with \cref{cor:linEthHard}, we find the following.
\thmcvpethhard*

\section{\texorpdfstring{$\ETH$}{ETH} hardness of \texorpdfstring{$\svp_{p,\gamma}$}{SVP p,gamma}}
\label{sec:svp}
In this section we show a reduction from $\sat$ in $n$ variables to $\svp_{p, \gamma}$ in a lattice of rank $\order{n}$. 
We do this by a reduction from $\cvp_{p, \gamma'}$ for a constant $\gamma'$ (an instance obtained in \cref{sec:cvp_hard}) to $\svp_{p, \gamma}$ for a constant $\gamma$. 
The result then follows from combining the reduction from $\sat$ to $\lin$ in \cref{thm:satToLinRed}, and from $\lin$ to $\cvp$ in \cref{thm:LinCvpReduction}. 
We first show that for $p > 2$, the integer lattice with the all half vector as the target forms a (family of) locally dense lattice gadgets with certain quantitative properties that will be useful in our reduction. 

\subsection{Locally Dense Integer Gadget}
\label{ssec:all_half_target}

We wish to show that in the integer lattice $\Z^n$, there are exponentially more vectors close to $\frac{1}{2}\vec 1_n$ than the number of {\em short} vectors. The following property of the theta function will be useful. 

\begin{theorem}\label{thm:thetaProp}
    For every $p \in (2, \infty)$, there always exists $\tau > 0$ such that 
    \begin{equation}
        \label{eqn:thetaProp}
        \Theta_p(\tau, 0) < \Theta_p(\tau, 1/2). 
    \end{equation}
\end{theorem}

\begin{proof}
	From the definition,
	$$
	\Theta_p(\tau, 0)=\sum_{z \in \mathbb{Z}} e^{-\tau|z|^p}=1+2 \sum_{z=1}^{\infty} e^{-\tau \cdot z^p},
	$$
	and
	$$
	\Theta_p(\tau, 1 / 2)=\sum_{z \in \mathbb{Z}} e^{-\tau|z-1 / 2|^p}=2 \sum_{z=1}^{\infty} e^{-\tau \cdot(z-1 / 2)^p} .
	$$
	It is straightforward to see $\Theta_p(\tau, 0)$ and $\Theta_p(\tau, 1 / 2)$ are absolutely convergent; therefore we can define the following function $f_p(\tau)$ on $\tau\in(0,\infty)$ that represents the difference between the two series, i.e.,
	$$
	f_p(\tau):=\Theta_p(\tau, 1/2)-\Theta_p(\tau, 0)=-1+2\sum_{z=1}^{\infty}\left[e^{-\tau(z-1/2)^{p}}-e^{-\tau z^{p}}\right]=-1+2\sum_{z=1}^{\infty}(-1)^{z-1}e^{-\tau (z/2)^p}.
	$$ 

    Let $p = 2 + \varepsilon$ for $\varepsilon > 0$, and write $f_{p}(\tau)$ as $$ f_{2+\varepsilon}(\tau)=-1+2\sum_{z=1}^{\infty}(-1)^{z-1}e^{-\tau (z/2)^{2+\varepsilon}}=-\sum_{z\in\Z}(-1)^{z}e^{-\tau |z/2|^{2+\varepsilon}}=-\sum_{z\in\Z}e^{(-{\tau/2^{2+\varepsilon}}) |z|^{2+\varepsilon}+\pi i z}.$$
    It suffices to prove that for all $\varepsilon>0$, there always exists $\tau>0$ such that $f_p(\tau)>0$.

    \textbf{Case 1: $\left(\varepsilon \in (0, 2)\right)$.}
    We will use the Poisson summation formula, which states that for any function $s(z) \in L^1(\mathbb{R})$, if its Fourier transform $\hat{s}(x) := \int_{-\infty}^{\infty} s(z) e^{-2 \pi i x z} d z \in L^1(\mathbb{R})$, then
    $$
    \sum_{z \in \mathbb{Z}} s(z)=\sum_{x \in \mathbb{Z}} \hat{s}(x).
    $$

    Define $C_{\tau, \varepsilon} := \tau / 2^{2+\varepsilon} > 0$ and $s_{\tau, \varepsilon}(z) := e^{-C_{\tau, \varepsilon} |z|^{2+\varepsilon} + \pi i z}$. In other words, $f_{2+\varepsilon}(\tau)=-\sum_{z\in\Z}s_{\tau,\varepsilon}(z)$. 
    Now we consider the Fourier transform of $s_{\tau, \varepsilon}(z)$, i.e., $$\hat{s}_{\tau, \varepsilon}(x):=\int_{-\infty}^{\infty} s_{\tau, \varepsilon}(z) e^{-2 \pi i x z} d z=\int_{-\infty}^{\infty} \exp\paren{-C_{\tau, \varepsilon}|z|^{2+\varepsilon}-2 \pi i (x-1 / 2)z} d z.$$

    Recall that the Fourier transform of $e^{-|z|^{2+\varepsilon}}$, $\hat{g}_{\varepsilon}(x) := \int_{-\infty}^{\infty} e^{-|z|^{2+\varepsilon}-2\pi i xz} dz$, has been well-studied in \cite{EOR91}. This motivates us to relate $\hat{s}_{\tau, \varepsilon}(x)$ to $\hat{g}_{\varepsilon}(x)$ by some translation and scaling as follows:
    $$
    \begin{aligned}
        \hat{s}_{\tau, \varepsilon}(x)&=\int_{-\infty}^{\infty} \exp\paren{-C_{\tau, \varepsilon}|z|^{2+\varepsilon}-2 \pi i (x-1 / 2)z} d z\\
        &=\int_{-\infty}^{\infty} \exp\paren{-|C_{\tau, \varepsilon}^{1/(2+\varepsilon)}z|^{2+\varepsilon}-2 \pi i \frac{x-1/2}{C_{\tau, \varepsilon}^{1/(2+\varepsilon)}}C_{\tau, \varepsilon}^{1/(2+\varepsilon)}z} d z\\
        &=\frac{1}{C_{\tau, \varepsilon}^{1/(2+\varepsilon)}}\int_{-\infty}^{\infty} \exp\paren{-|C_{\tau, \varepsilon}^{1/(2+\varepsilon)}z|^{2+\varepsilon}-2 \pi i \frac{x-1/2}{C_{\tau, \varepsilon}^{1/(2+\varepsilon)}}C_{\tau, \varepsilon}^{1/(2+\varepsilon)}z} d\left(C_{\tau, \varepsilon}^{1/(2+\varepsilon)}z\right)\\
        &=\frac{1}{C_{\tau, \varepsilon}^{1/(2+\varepsilon)}}\int_{-\infty}^{\infty} \exp\paren{-|z|^{2+\varepsilon}-2 \pi i \frac{x-1/2}{C_{\tau, \varepsilon}^{1/(2+\varepsilon)}}z} d z=\frac{1}{C_{\tau, \varepsilon}^{1/(2+\varepsilon)}}\cdot\hat{g}_{\varepsilon}\left(\frac{x-1/2}{C_{\tau, \varepsilon}^{1/(2+\varepsilon)}}\right).
    \end{aligned}
    $$
    To be eligible for the Poisson summation formula $\sum_{z \in \mathbb{Z}} s_{\tau,\varepsilon}(z)=\sum_{x \in \mathbb{Z}} \hat{s}_{\tau,\varepsilon}(x)$, we need to show that for any $\tau > 0$ and $\varepsilon > 0$, $s_{\tau, \varepsilon}, \hat s_{\tau, \varepsilon} \in L^1(\R)$, which is given by the following lemmas.
    \begin{lemma}\label{lem:1}
        For any $\tau > 0$ and $\varepsilon > 0$, $s_{\tau, \varepsilon}\in L^1(\R)$.   
    \end{lemma}

    \begin{proof}
        Note that 
    	$$
    	\begin{aligned}
    		\int_{-\infty}^{\infty} |s_{\tau,\varepsilon}(z)|dz=&\int_{-\infty}^{\infty} \left| e^{-C_{\tau, \varepsilon} |z|^{2+\varepsilon}+\pi i z} \right|dz\\
    		=&2\int_{0}^{\infty}  e^{-C_{\tau, \varepsilon} z^{2+\varepsilon}}dz\\
    		=&2\left(\int_{0}^{1}  e^{-C_{\tau, \varepsilon} z^{2+\varepsilon}}dz+\int_{1}^{\infty}  e^{-C_{\tau, \varepsilon} z^{2+\varepsilon}}dz\right)\\
    		\leq& 2\left(\int_{0}^{1}  e^{-C_{\tau, \varepsilon} z^{2+\varepsilon}}dz+\int_{1}^{\infty}  e^{-C_{\tau, \varepsilon} z^{2}}dz\right).
    	\end{aligned}
    	$$
    	Since $e^{-C_{\tau, \varepsilon}z^{2+\varepsilon}}$ is bounded and continuous on $[0,1]$, and $e^{-C_{\tau, \varepsilon} z^{2}}$ is a Gaussian function, we can conclude that both integrals are finite, which completes the proof. 
    \end{proof}

    \begin{lemma}\label{lem:2}
        For any $\tau > 0$ and $\varepsilon > 0$, $\hat s_{\tau, \varepsilon}\in L^1(\R)$.  
    \end{lemma}

    \begin{proof}
        By \cite{Fol09}, for any continuous and piecewise smooth function $f$, if both $f\in L^2$ and its first derivative  $f^{\prime}\in L^2$, then its Fourier transform $\hat{f}\in L^1$. Since $s_{\tau,\varepsilon}(z)$ is clearly smooth on $z\in(-\infty,0)\cup(0,\infty)$ and continuous at $z=0$ as $\lim_{z\rightarrow0}s_{\tau,\varepsilon}(z)=s_{\tau,\varepsilon}(0)=1$, it suffices to show that $s_{\tau,\varepsilon}\in L^2(\R)$ and $s^{\prime}_{\tau,\varepsilon}\in L^2(\R)$.
	
    	For $s_{\tau,\varepsilon}\in L^2(\R)$, we argue in a similar way to the proof of \cref{lem:1}. Note that,
    	$$
    	\begin{aligned}
    		\int_{-\infty}^{\infty} |s_{\tau,\varepsilon}(z)|^2 dz=&\int_{-\infty}^{\infty} \left| e^{-C_{\tau, \varepsilon} |z|^{2+\varepsilon}+\pi i z} \right|^2 dz\\
    		=&2\int_{0}^{\infty}  e^{-2C_{\tau, \varepsilon} z^{2+\varepsilon}}dz\\
    		=&2\left(\int_{0}^{1}  e^{-2C_{\tau, \varepsilon} z^{2+\varepsilon}}dz+\int_{1}^{\infty}  e^{-2C_{\tau, \varepsilon} z^{2+\varepsilon}}dz\right)\\
    		\leq& 2\left(\int_{0}^{1}  e^{-2C_{\tau, \varepsilon} z^{2+\varepsilon}}dz+\int_{1}^{\infty}  e^{-2C_{\tau, \varepsilon} z^{2}}dz\right).
    	\end{aligned}
    	$$
    	Since $e^{-2C_{\tau, \varepsilon} z^{2+\varepsilon}}$ is bounded and continuous on $[0,1]$, and $e^{-2C_{\tau, \varepsilon} z^{2}}$ is a Gaussian function, we can conclude that both integrals are finite, which completes the proof of $s_{\tau,\varepsilon}\in L^2(\R)$. 
    	
    	
    	Now we calculate the first derivative of $s_{\tau,\varepsilon}(z)$, i.e.,
    	$$
    	s^{\prime}_{\tau,\varepsilon}(z) = e^{-C_{\tau, \varepsilon} |z|^{2+\varepsilon}+\pi i z} \cdot\left(-C_{\tau, \varepsilon}(2+\varepsilon) |z|^{1+\varepsilon}\operatorname{sgn}(z)+\pi i\right),
    	$$
    	where
    	$$
    	\operatorname{sgn}(z) = \left\{
    	\begin{aligned}
    		1, \quad & z>0\\
    		0, \quad & z=0\\
    		-1, \quad & z<0
    	\end{aligned}
    	\right.
    	$$
    	Note that this formula works for all $z\in\R$. In particular, one can verify that 
    	$$
    	\lim_{z\rightarrow0}\frac{s_{\tau,\varepsilon}(z)-s_{\tau,\varepsilon}(0)}{z} = \lim_{z\rightarrow0}\frac{e^{-C_{\tau, \varepsilon} |z|^{2+\varepsilon}+\pi i z}-1}{z} = \lim_{z\rightarrow0}\frac{-C_{\tau, \varepsilon} |z|^{2+\varepsilon}+\pi i z}{z} = \pi i = s^{\prime}_{\tau,\varepsilon}(0).
    	$$
    	Then it is straightforward to see that
    	$$
    	\begin{aligned}
    		\int_{-\infty}^{\infty} |s^{\prime}_{\tau,\varepsilon}(z)|^2 dz=&\int_{-\infty}^{\infty} \left| e^{-C_{\tau, \varepsilon} |z|^{2+\varepsilon}+\pi i z} \right|^2 \cdot \left(\sqrt{ C_{\tau, \varepsilon}^2(2+\varepsilon)^2 |z|^{2+2\varepsilon} + \pi^2} \right)^2 dz\\
    		=&2\int_{0}^{\infty}  e^{-2C_{\tau, \varepsilon} z^{2+\varepsilon}} \cdot \left(C_{\tau, \varepsilon}^2(2+\varepsilon)^2 z^{2+2\varepsilon} + \pi^2 \right) dz\\
    		=&2\int_{0}^{1}  e^{-2C_{\tau, \varepsilon} z^{2+\varepsilon}} \cdot \left(C_{\tau, \varepsilon}^2(2+\varepsilon)^2 z^{2+2\varepsilon} + \pi^2 \right) dz\\
            & \hspace{2pt} +2\int_{1}^{\infty}  e^{-2C_{\tau, \varepsilon} z^{2+\varepsilon}} \cdot \left(C_{\tau, \varepsilon}^2(2+\varepsilon)^2 z^{2+2\varepsilon} + \pi^2 \right) dz\\
    		\leq& 2\left(\int_{0}^{1}  e^{-2C_{\tau, \varepsilon} z^{2+\varepsilon}} \cdot \left(C_{\tau, \varepsilon}^2(2+\varepsilon)^2 z^{2+2\varepsilon} + \pi^2 \right) dz+\int_{1}^{\infty}  C_{\tau, \varepsilon}^{\prime}z^{-2} dz\right),
    	\end{aligned}
    	$$
    	where $C_{\tau, \varepsilon}^{\prime}$ is a finite constant depending on $\tau$ and $\varepsilon$. This is because an exponential function always decays faster than any polynomial function, which implies that for any given $\tau>0$ and $\varepsilon>0$, we can find $C_{\tau, \varepsilon}^{\prime}>0$ such that $C_{\tau, \varepsilon}^{\prime}z^{-2}\geq e^{-2C_{\tau, \varepsilon} z^{2+\varepsilon}} \cdot (C_{\tau, \varepsilon}^2(2+\varepsilon)^2 z^{2+2\varepsilon} + \pi^2 )$ for all $z\geq 1$. Since the integrand of the first integral is bounded and continuous on $[0,1]$, and the second integral $\int_{1}^{\infty}  C_{\tau, \varepsilon}^{\prime}z^{-2} dz=C_{\tau, \varepsilon}^{\prime}$, we can conclude that both integrals are finite, which completes the proof of $s^{\prime}_{\tau,\varepsilon}\in L^2(\R)$.
    \end{proof}

    Recall from \cref{lem:eor} that $$\hat{g}_{\varepsilon}(x) \sim 2 \sum_{m=1}^{\infty} \frac{(-1)^{m+1}}{m!} \sin \left(\frac{m \pi (2+\varepsilon)}{2}\right) \Gamma(m (2+\varepsilon)+1)(2 \pi|x|)^{-m (2+\varepsilon)-1},$$ 
    and observe that for any $x \in \R$ and $\varepsilon \in (0, 2)$, the first term in the asymptotic expansion of $\hat{g}_{\varepsilon}(x)$ is negative as $\sin(\pi (2+\varepsilon) / 2) < 0$.
    By \cref{def:asymptotic}, there exists $N_0 \in \R$ such that for all $|x| > N_0$, $\hat{g}_{\varepsilon}(x) < 0$. 
    Therefore, we can always choose $\tau_0 \in (0, N_0^{-(2+\varepsilon)})$, which satisfies that for all $x \in \Z$, $$N_0 < \frac{1}{\tau_0^{1/(2+\varepsilon)}} = \frac{1}{2 C_{\tau_0, \varepsilon}^{1/(2+\varepsilon)}} \leq \abs{\frac{x - 1/2}{C_{\tau_0, \varepsilon}^{1/(2+\varepsilon)}}}.$$
    Since $C_{\tau_0, \varepsilon}^{1/(2+\varepsilon)}>0$, we have that for all $x \in \Z$, 
    $$
    \hat s_{\tau_0, \varepsilon}(x) = \frac{1}{C_{\tau_0, \varepsilon}^{1/(2+\varepsilon)}}\cdot\hat{g}_{\varepsilon}\left(\frac{x-1/2}{C_{\tau_0, \varepsilon}^{1/(2+\varepsilon)}}\right) < 0.
    $$
    By \cref{lem:1,lem:2}, we can apply the Poisson summation formula to conclude that $$f_{2+\varepsilon}(\tau_0) = -\sum_{z \in \Z} s_{\tau_0, \varepsilon}(z) = -\sum_{x \in \Z} \hat s_{\tau_0, \varepsilon}(x) > 0.$$

This concludes the proof of \textbf{Case 1}.

    \textbf{Case 2: $(\varepsilon \geq 2)$.} In this case, we simply choose $\tau_0 = 1$, which yields that 
    $$
    f_{2+\eps}(1)=f_p(1)=-1+2\sum_{z =1}^{\infty} \left(e^{-(z-1/2)^p}-e^{- z^p}\right).$$
    Since $e^{- (z-1/2)^p}-e^{-z^p}>0$ for all $z \in \Z^+$ and $p\geq4$, it follows that 
    $$
    f_p(1) > -1 + 2(e^{-1/2^p} - e^{-1}) \geq -1 + 2(e^{-1/2^4} - e^{-1}) > 0,
    $$
    which completes our proof.
\end{proof}

\begin{theorem}
    \label{thm:SvpIntegerGadget}
    For any $p \in (2, \infty),~\sigma > 1$, there exists constants $\delta \in (0, 1/2)$, $\phi_0, \phi_1 > 1$ and $C_r>0$, such that for any $n \in \Z^+$, if $r :=C_r \cdot {n^{1/p}}$ then the following holds for the integer lattice. 
    \begin{align}
        \nump{\Z^{n}, (1-\delta)^{1/p} r, \frac{1}{2}\vec{1}_{n}} \geq \max\bigg\{ &\phi_0^{n - o(n)} \nump{\Z^{n}, r, \vec{0}}, \notag \\
        &\phi_1^{n - o(n)} \nump{\Z^{n}, (1-\delta\sqrt{\sigma})^{1/p} r, \frac{1}{2} \vec{1}_{n}} \bigg\}.
        \label{eqn:SvpIntegerGadgetCond}
    \end{align} 
     The constants $\delta, \phi_0, \phi_1, C_r$ only depend  upon $p$ and can be efficiently computed with required precision, from $p$. 
\end{theorem}

\begin{proof}
Fix any $p$ and $\sigma$. Fix $\tau > 0$ as in \cref{thm:thetaProp}. Then define 
$$\rho:=\frac{\Theta_p(\tau,1/2)}{\Theta_p(\tau,0)}>1.$$ 
Notice that the precision to which we would approximate $\rho$ would be independent of the input size, so such a $\rho$ is efficiently computable in constant time. 
For a $\delta \in (0, 1/2)$ to be determined in what follows, fix 
$$C_r := \frac{\mu_p(\tau,1/2)^{1/p}}{(1-\delta)^{1/p}}.$$
 By \cref{thm:theta_gives_good_approx}, there exists a constant $C^*>0$ such that,

\begin{equation}
\label{eq:bd1}
\nump{\Z^{n},(1-\delta)^{1/p}\cdot r,\frac{1}{2}\vec{1}_{n}} \geq \exp\left(-C^* \sqrt{n}\right)\cdot \exp(\tau n \cdot \mu_p(\tau,1/2))\cdot \Theta_p(\tau , 1/2)^{n}.
\end{equation}

By \cref{lemma:theta0}, 
\begin{equation}
\label{eq:bd2}
\nump{\Z^{n},r,\vec{0}} \leq \exp\paren{\tau \cdot\frac{n ~\mu_p(\tau,1/2)}{1-\delta}}\cdot \Theta_p(\tau,0)^{n}.
\end{equation}
From \cref{eq:bd1,eq:bd2}, we get that 
\[\frac{\nump{\Z^{n},(1-\delta)^{1/p}\cdot r,\frac{1}{2}\vec{1}_{n}}}{\nump{\Z^{n},r,\vec{0}}} \geq \rho^{n} \cdot \exp\left(-C_r^p\tau\delta n -o(n)\right).\]
Because $\rho>1$, we can fix $\delta>0$ such that,
$$\delta < \frac{\log \rho}{\tau \mu_p(\tau, 1/2) + \log \rho},$$
which implies that $\rho \cdot \exp(-\tau \delta C_r^p)>1$. 
Then set $\phi_0:=\rho \cdot \exp(-\tau \delta C_r^p)$. This implies that, 
$$\frac{\nump{\Z^{n},(1-\delta)^{1/p}\cdot r,\frac{1}{2}\vec{1}_{n}}}{\nump{\Z^{n},r, \vec{0}}} \geq \phi_0^{n -o(n)}.$$

For the second inequality, again use \cref{lemma:theta0} to get,
\begin{equation}
\label{eq:bd3}
N_p\left(\Z^{ n},(1-\delta\sqrt{\sigma})^{1/p}\cdot r,\frac{1}{2}\vec{1}_{n}\right)\leq \Theta_p(\tau,1/2)^{n} \cdot \exp\paren{\tau \cdot n C_r^p \cdot (1-\delta\sqrt{\sigma})}.
\end{equation}
Then set  $\phi_1:=\exp \paren{\tau\delta C_r^p\cdot(\sqrt{\sigma}-1)}>1$ which implies,
\[\frac{N_p\left(\Z^{n},(1-\delta)^{1/p}\cdot r,\frac{1}{2}\vec{1}_{n}\right)}{N_p \left(\Z^{ n},(1-\delta\sqrt{\sigma})^{1/p}\cdot r,\frac{1}{2}\vec{1}_{n} \right)}\geq \exp \paren{\tau\delta n C_r^p\cdot(\sqrt{\sigma}-1) -o(n)}=\phi_1^{n-o(n)}. \]
\end{proof}

\subsection{From \texorpdfstring{$\lin_\varepsilon$ to $\svp_{p, \gamma}$}{lin to svp}}
In the following, we  start with a $\cvp_{p, \gamma'}$ instance with the special property that twice the target vector is in the lattice, and find a lattice $\cL$ and a parameter $r$ such that the number of vectors in $\cL$ of length at most $r$ are $2^{\Omega(n)}$ times more for a YES instance than for a NO instance. 
\begin{lemma}
    \label{lem:cvpToSvp}
    For any $p \in [1, \infty)$, suppose $(B', \vec t', r')$ is a $\cvp_{p, \gamma'}$ instance, such that $B' \in \Z^{m \times m'}, \vec t' \in \lat(B') / 2$ and $\vec t' \in \Z^m$. Then for any $d \in \Z^+ $, given a lattice $\ldgr\subseteq \R^d$ with basis $B^\dagger \in \R^{d\times d'}$, and a target $\tdgr \in \R^d$, there exists an efficiently computable matrix $B$ that generates a lattice $\lat(B)$ in $m + d + 1$ dimensions such that for any constants $\gamma'>1$ and any radii $r_G$,$r_A$,
    \begin{enumerate}
        \item If $\di_p\paren{\vec{t}',\cL(B')}\leq r'$  then,
        \begin{equation}
            \label{eqn:cvpToSvpYes}
        \nump{\lat, r_G, \vec 0} \geq \nump{\ldgr, (r_G^{p} - 1 - r'^p)^{1/p},\tdgr},
    \end{equation}
        \item If  $\di_p\paren{\vec{t}',\cL(B')}> \gamma'r'$ then,
        \begin{align}
            \label{eqn:cvpToSvpNo}
        \nump{\lat, r_A, \vec 0} \leq(r_A +4) \cdot \nump{\Z^m, r_A, \vec 0} \cdot \bigg(&\max_{\cl_1 \in \Z}\left\{\nump{\ldgr,~(r_A^p - (\gamma' r')^p)^{1/p}, \cl_1\cdot \tdgr} \right\} + \notag\\ 
        &\max_{\cl_0 \in 2\Z}\left\{  \nump{\ldgr, r_A, \cl_0 \cdot \vec{t}^\dagger} \right\}\bigg).
    \end{align}
    \end{enumerate}
\end{lemma}

\begin{proof}
    
    Define
    \begin{equation}
        \label{eqn:cvpToSvpBasis}
    B := \begin{pmatrix}
        B' & 0 & -\vec t' \\
        0 & \dgr{B} & - \vec{t}^{\dagger} \\
        0 & 0 & 1
    \end{pmatrix}. 
    \end{equation}
    Then $\lat = \lat(B)$ is a lattice in $m + d + 1$ dimensions, and $B$ is clearly efficiently computable. 
    Suppose that the input $\cvp_{p, \gamma'}$ instance was a \yes instance. This implies that $\dist(\cL', \vec t') \leq r'$. Then $\exists \vec x \in \Z^{m'}$ such that $\norm{B'\vec x - \vec t'}_p \leq r'$. 
    For any $r_G$, consider a vector $\vec v \in \ballp{\ldgr, (r_G^p - 1 - r'^p)^{1/p}, \tdgr}$ such that $\vec v = \dgr B \vec y$ for some $\vec y \in \Z^{d'}$. Then the vector $\vec u := B \cdot (\vec x, \vec y, 1)$ satisfies
    $$ \normpp{\vec u} \leq \normpp{B'\vec x - \vec t} + \normpp{\dgr B \vec y - \tdgr} + 1 \leq r_G^p. $$
    Therefore, $\vec u \in \ballp{\lat, r_G, \vec 0}$. Since this mapping between $\vec v$ and $\vec u$ is injective, we have that
    $$ \nump{\lat, r_G, \vec 0} \geq \nump{\dgr B, (r_G^{p} - 1 - r'^p)^{1/p}, \tdgr}. $$
    This completes the proof of item 1. 
    Next, suppose that the input $\cvp_{p, \gamma'}$ instance was a \no instance. This implies that $\dist(\cL', \vec t') > \gamma' r'$. 
    For any $r_A$, let $\vec u \in \ballp{\lat, r_A, \vec 0}$ such that $\vec u = (B'\vec x, B^\dagger \vec y , \cl)$ for some $\vec x \in \Z^{m'},~\vec y\in\Z^{d'}$ and $\cl\in \Z$. We have 
    \[ \normpp{\vec u} \leq \normpp{B'\vec x - \cl\vec t'} + \normpp{\dgr B \vec y - \cl\tdgr} + \cl^p \leq r_A^p, \]
    which implies that 
    \[ \normpp{B'\vec x - \cl \vec t'} \leq r_A^p - \cl^p \leq r_A^p; \]
    \begin{equation}
    \label{eqn:coeffForGadget}
     \normpp{\dgr B \vec y - \cl \vec t^{\dagger}} ~\leq ~r_A^p - \cl^p - (\dist_p(\lat', \ell \vec t'))^p ~\leq ~r_A^p - (\dist_p(\lat', l \vec t'))^p.
     \end{equation}
    Since $\lat' \subseteq \Z^m$, and $\cl \vec t' \in \Z^m$, we have that the number of values that $B'\vec x$ can take is 
    \begin{align}
        \nump{\lat', r_A, \cl \vec t'} \notag &\leq \nump{\Z^m, r_A, \cl \vec t'} && (\because \lat' \subseteq \Z^m) \notag \\
        &= \nump{\Z^m, r_A, \vec 0} && (\because \cl\vec t' \in \Z^m) \label{eqn:svptemp1}.
    \end{align}
    
    Let $S_\cl$ be the set of points in $\lat$ of length at most $r_A$ such that their last coordinate is $\cl$. Observe that $\abs{S_\cl} = \abs{S_{-\cl}}$. 
    Then, 
    $$ \nump{\lat, r_A, \vec 0} \leq \sum_{\cl=-\floor{r_A}}^{\floor{r_A}} \abs{S_\cl} \leq 2 \sum_{\cl = 0}^{\floor{r_A}} \abs{S_\cl}. $$

    \textbf{Suppose that $\cl$ is even.} Then $\vec t' \in \lat' / 2 \implies \dist_p(\lat', \cl\vec t') = 0$. Thus, we have that the number of possible values $\vec B^\dagger y$ can take is at most,
    \begin{align}
        \label{eqn:svptemp2}
        \nump{\ldgr, r_A, \cl \tdgr}.
    \end{align}
    The number of possible vectors $\vec u$ in this case is at most the product of \cref{eqn:svptemp1} and \cref{eqn:svptemp2},
    \begin{align}
        \label{eqn:svpslodd}
        \abs{S_l} \leq \nump{\Z^m, r_A, \vec 0} \nump{\ldgr, r_A, \cl\vec {t}^\dagger}.
    \end{align}
    \textbf{Suppose that $\cl$ is odd.} Then $\vec t' \in \lat' / 2 \implies \cl\vec t' \in \lat' + \vec t'$. Let $\cl = 2z+1$. We have that
    $$\dist_p(\lat', (2z+1)\vec t') = \dist_p(\lat', \vec v' + \vec t') = \dist_p(\lat', \vec t'),$$
    where $\vec v' \in \lat'$. 
    Therefore, the number of possible values $B^\dagger \vec y$ can take while satisfying \cref{eqn:coeffForGadget} is at most 
    \begin{align}
        \label{eqn:svptemp3}
        \nump{\ldgr, (r_A^p - (\gamma' r')^p)^{1/p}, \cl \tdgr}.
    \end{align}
    The number of possible vectors $\vec u$ in this case is at most the product of \cref{eqn:svptemp1} and \cref{eqn:svptemp3},
    \begin{align}
        \label{eqn:svpsleven}
        \abs{S_l} \leq \nump{\Z^m, r_A, \vec 0} \nump{\ldgr, (r_A^p - (\gamma' r')^p)^{1/p}, \cl\tdgr}.
    \end{align}
    Thus summing over even and odd $\cl$ from \cref{eqn:svpslodd} and \cref{eqn:svpsleven} we get an upper bound on $\nump{\lat, r_A, \vec 0}$ as,
    \begin{align*}
        &\leq  2\cdot \nump{\Z^m, r_A, \vec 0} { \sum_{j=0}^{\lceil r_A/2 \rceil}\paren{\nump{\ldgr,~(r_A^p - (\gamma' r')^p)^{1/p}, (2j+1)\cdot \tdgr} + \nump{\ldgr, r_A, (2j) \cdot \vec{t}^\dagger} }}\\
        &\leq(r_A +4) \cdot \nump{\Z^m, r_A, \vec 0} \cdot \paren{\max_{\cl_1 \in \Z}\left\{\nump{\ldgr,~(r_A^p - (\gamma' r')^p)^{1/p}, \cl_1\cdot \tdgr} \right\} + \max_{\cl_0 \in 2\Z}\left\{  \nump{\ldgr, r_A, \cl_0 \cdot \vec{t}^\dagger} \right\}}
    \end{align*}
    This completes the proof of item 2.
\end{proof}



This gives us the following lemma, that says that for any $\lin_\varepsilon$ instance, we can efficiently compute a lattice such that the number of short vectors in the \yes case are exponentially more than the number of short vectors in the \no case. 

\begin{lemma}
    \label{lem:AGRelation}
    For all $p \in (2,\infty)$ , there is an efficient algorithm that takes as an input a $\lin_\varepsilon$ instance $(M, \vec v)$ over $\F_2$ in $n$ variables and $m = \order{n}$ equations, and outputs $(B, r, \gamma, A, G)$ where $B$ generates a lattice $\lat$ in $\order{n}$ dimensions, $\gamma > 1$ is a constant, $r > 0$ is a radius, and constants $A,~G$ are such that $G \geq 2^m A$ and 
    \begin{enumerate}
        \item if it was a \yes instance of $\lin_\varepsilon$ then, $\nump{\lat, r, \vec 0} \geq G$;
        \item if it was a \no instance of $\lin_\varepsilon$ then, $\nump{\lat, \gamma r, \vec 0} \leq A$.
    \end{enumerate}
\end{lemma}

\begin{proof}
    Fix a $p$. We describe the algorithm followed by correctness. 

    \textbf{Algorithm.}
    Define $\gamma' := \paren{1 + \frac{8\varepsilon}{3}}^{1/p}$.
    On input a $\lin_\varepsilon$ instance $M \in \F_2^{m \times n},~\vec v \in \F_2^m$, the reduction computes a $\cvp_{p, \gamma'}$ instance $(B', \vec t', r')$ using the algorithm from \cref{thm:LinCvpReduction}.
    Recall from the proof of \cref{thm:LinCvpReduction} that if this instance is a \yes instance, then $\dist_p(\lat(B'), \vec t') \leq r'$, and if this instance is a \no instance then $\dist_p(\lat(B'), \vec t') \geq (\gamma' r')$. 
    Also recall that by construction in \cref{thm:LinCvpReduction}, $B' \in \Z^{m \times (n+m)}$, $r' = {(3m/8)^{1/p}}$, $2\Z^m \subseteq \lat(B') \subseteq \Z^m$ and $\vec t' \in \lat(B') / 2$ as $\vec t' \in \Z^m$.
    Therefore, this instance satisfies the conditions for \cref{lem:cvpToSvp} to hold. 
    The reduction then applies the transformation in \cref{lem:cvpToSvp} with $r_G=r$, $ r_A= \gamma r$, $B^\dagger= \alpha I_d$ and $\tdgr=(\alpha/2) \cdot \vec{1}_d$ on this $\CVP_{p,\gamma}$ instance, for a particular choice of efficiently computable parameters $\alpha$, $r$, $\gamma$ and $d$ (to be determined in what follows), and gets a generating set $B$ of a lattice $\lat$ in $m + d + 1$ dimensions. 
    Define,
    \begin{equation}
        \label{eqn:svpGood}
    G := \nump{\ldgr, (r^{p} - 1 - r'^p)^{1/p}, \tdgr},
\end{equation}
    \begin{align}
        \label{eqn:svpAnnoying}
        A := (\gamma r+4) \cdot \nump{\Z^m, \gamma r, \vec 0} \cdot \bigg(&\max_{\cl_1 \in \Z}\left\{\nump{\ldgr,~((\gamma r)^p - (\gamma' r')^p)^{1/p}, \cl_1\cdot \tdgr} \right\} + \notag \\
        &\max_{\cl_0 \in 2\Z}\left\{  \nump{\ldgr, \gamma r, \cl_0 \cdot \vec{t}^\dagger} \right\}\bigg). 
    \end{align}
    The reduction outputs $(B, \gamma, r, A, G)$. Efficiency is clear, as $A$ and $G$ can be approximated efficiently using the theta function. We will now prove the correctness. 

    \textbf{Correctness.}
    Notice that $G$ and $A$ are simply the RHS of \cref{eqn:cvpToSvpYes} and \cref{eqn:cvpToSvpNo} respectively. \Cref{eqn:svpGood} and \cref{eqn:cvpToSvpYes} imply that if the input was a \yes instance of $\lin_\varepsilon$, then $$\nump{\lat, r, \vec 0} \geq G.$$
    Similarly, \cref{eqn:svpAnnoying} and \cref{eqn:cvpToSvpNo} imply that if the input was a \no instance of $\lin_\varepsilon$, then $$\nump{\lat, \gamma r, \vec 0} \leq A.$$
    Let $\delta$ and $C_r$ be as in \cref{thm:SvpIntegerGadget} for the parameters $(p, \sigma=\gamma'^p)$. Fix $r, \gamma$ such that 
    $$r^p = \svprval; \quad \quad \gamma^p = \svpgammaval.$$
    By \cref{lemma:theta0}, for any $\gamma r,~\tau$, there exists a $K > 0$ which is independent of $m$, such that 

    \begin{align}
        \nump{\Z^{m}, \gamma r, \vec 0} &\leq \exp(\tau (\gamma r)^p) \Theta(\tau, \vec 0) \notag \\ 
        &= \exp(\tau (\gamma r)^p) \Theta(\tau, 0)^{m} \notag \\ 
        &\leq K^{m} && (\because r^p \in \order{m}).
    \label{eqn:numpKsvp}
    \end{align}

    Let $\phi_0, \phi_1$, be the constants from \cref{thm:SvpIntegerGadget} for the particular $p$.
    Set $$d = \max\left\{ \ceil{\log_{\phi_0}(3K)}, \ceil{\log_{\phi_1}(3K)} \right\} m,$$ and let $~\dgr r = C_r \cdot(d)^{1/p}$. Note that $d = \order{m}$.
    A simple but tedious calculation shows that the following is true.
    \begin{restatable}{claim}{claimSvpGadgetParam}
        \label{claim:SvpGadgetParameter}
        For any $p,\gamma', r'$, if we set 
        \begin{align*}
        \alpha^p = \svpalphaval;
        \quad \quad r^p = \svprval; 
        \quad \quad \gamma^p = \svpgammaval;
    \end{align*}
    then the following holds. 
    \begin{enumerate}
        \item $\left(r^{p} - 1 - r'^p\right)^{1/p} \geq (1-\delta)^{1/p} (\alpha r^{\dagger})$. \label{im:svpRadii1}
        \item $\gamma r \leq \alpha r^{\dagger}$ \label{im:svpRadii2}
        \item $\left((\gamma r)^p - (\gamma' r')^p \right)^{1/p}\leq (1 - \delta\sqrt{\gamma'^p})^{1/p} (\alpha r^{\dagger})$ \label{im:svpRadii3}
    \end{enumerate}
    \end{restatable}

    To see \cref{im:svpRadii1} notice that,
    \begin{align*}
        r^p-1-r'^p&= \left(1-\frac{\delta}{2} \right)\frac{ 2r'^p}{\delta} - r'^p =2\frac{(1-\delta)}{\delta}r'^p= (1-\delta)(\alpha r^\dagger)^p.
    \end{align*}
    To see \cref{im:svpRadii2}, notice that since $r'^p=\Omega({m})$, we can assume that $(1+\delta/100)\leq r'^p/10$. This implies,
    \begin{align*}
        (\gamma r)^p &\leq \left(1+\delta/100\right)\left(1+\left(1-\frac{\delta}{2}\right)\frac{2r'^p}{\delta}\right)\\
        &\leq \frac{ r'^p}{10}+\left(1+\delta/100\right)\left(\left(1-\frac{\delta}{2}\right)\frac{2r'^p}{\delta}\right)\\
        & \leq \paren{\frac{\delta}{10}+\paren{1+\frac{\delta}{100}}\paren{2-\delta}}\left(\frac{r'^p}{\delta}\right)\\
        &=\paren{2-\frac{22}{25}\delta-\frac{\delta^2}{100}}\paren{\frac{r'^p}{\delta}}\\
        &< \paren{\frac{2r'^p}{\delta}}\\
        &= (\alpha r^\dagger)^{p}.
    \end{align*}
    For \cref{im:svpRadii3}, use the fact $\gamma^p \leq 1+\delta\cdot \frac{(\sqrt{\gamma'}-1)^2}{2}$ and $\gamma'\geq 1$. Again, since $r'^p=\Omega({m})$ , without loss of generality we can assume that, $\frac{\delta\paren{\sqrt{\gamma'^p}-1}^2 r'^p}{2}\geq \gamma^p$. Hence we get,
    \begin{align*}
        (\gamma r)^p &\leq \paren{\frac{2r'^p}{\delta}}\paren{\paren{1-\frac{\delta}{2}}\paren{1+\delta\cdot \frac{(\sqrt{\gamma'^p}-1)^2}{2}}}+\gamma^p\\
        &= \paren{\frac{2r'^p}{\delta}} \paren{\paren{1-\delta\sqrt{\gamma'^p}}+\paren{\frac{\delta \gamma'^p}{2}+\frac{\delta^2 \sqrt{\gamma'^p
        }}{2}}-\paren{\frac{\delta(\delta+\delta \gamma'^p)}{4}}}+\gamma^p\\
        &= \paren{\frac{2r'^p}{\delta}} \paren{\paren{1-\delta\sqrt{\gamma'^p}} + \frac{\delta \gamma'^p}{2} -\frac{\delta^2\paren{\sqrt{\gamma'^p}-1}^2}{4}}+\gamma^p\\
        &=(1 - \delta\sqrt{\gamma'^p}) (\alpha r^{\dagger})^p+(\gamma'r')^p- \frac{\delta\paren{\sqrt{\gamma'^p}-1}^2 r'^p}{2}+\gamma^p\\
        &\leq (1 - \delta\sqrt{\gamma'^p}) (\alpha r^{\dagger})^p+(\gamma'r')^p .
    \end{align*}
This concludes the proof of the claim.
\Cref{im:svpRadii1} implies that 
        \begin{align}
        \label{eqn:GZBound}
        G = \nump{\alpha \Z^{d}, (r^{p} - 1 - r'^p)^{1/p}, {(\alpha/2)}\cdot \vec{1}_{d}} \geq \nump{\alpha \Z^{d}, (1-\delta)^{1/p}\alpha r^{\dagger}, {(\alpha/2)} \cdot\vec{1}_{d}}.
        \end{align}
        \Cref{im:svpRadii2,im:svpRadii3} and \cref{eqn:numpKsvp} imply that
        \begin{align*}
        \label{eqn:AZbound}
            A &= (\gamma r+4) \cdot \nump{\Z^m, \gamma r, \vec 0} \cdot \bigg(\max_{\cl_1 \in \Z}\left\{\nump{\ldgr,~((\gamma r)^p - (\gamma' r')^p)^{1/p}, \cl_1\cdot \tdgr} \right\} +  \max_{\cl_0 \in 2\Z}\left\{  \nump{\ldgr, \gamma r, \cl_0 \cdot \vec{t}^\dagger} \right\}\bigg) \notag \\
            &\leq  (\gamma r+4) K^m  \bigg( \max_{\cl_1 \in \Z} \left\{\nump{\alpha \Z^{d}, (1 - \delta\sqrt{\gamma'^p})^{1/p} \alpha r^{\dagger}, {(\alpha/2)}{\cl_1} \cdot \vec{1}_{d}}\right\}+ \max_{\cl_0 \in 2\Z} \left\{ \nump{\alpha \Z^{d}, \alpha r^{\dagger}, {(\alpha/2)}{\cl_0} \cdot \vec{1}_{d}} \right\}\bigg) \\
            &\leq  (\gamma r+4)K^m \paren{\dfrac{1}{\phi_1^{d}} + \dfrac{1}{\phi_0^{d}}}\nump{\alpha \Z^{d}, (1-\delta)^{1/p}\alpha r^{\dagger}, {(\alpha/2)} \cdot\vec{1}_{d}}\\
            &\leq (\gamma r +4)\paren{\frac{2}{3^m}}\cdot G
        \end{align*}
        For our choice of $d, \dgr r, \delta$, together with \cref{eqn:SvpIntegerGadgetCond}, and the fact that $r = \order{m^{1/p}}$, these imply that for a sufficiently large $m$, $G \geq 2^m A$.
\end{proof}


\begin{theorem}
    \label{thm:linToSvp}
    For any $p \in (2, \infty)$, there exists a constant $\gamma > 1$ such that for all $n \in \Z^+$, there is a polynomial time randomized Karp reduction from $\lin_\varepsilon$ in $n$ variables and $\order{n}$ equations to $\svp_{p, \gamma}$ in a lattice of rank $\order{n}$.
\end{theorem}

\begin{proof} Fix a $p$. We first describe the algorithm, and then show correctness. 

    \textbf{Algorithm.}
The reduction uses the algorithm from \cref{lem:AGRelation} to compute $(B, \gamma, r, A, G)$. Then it
    samples a prime number $$q \in [\sqrt{AG} / 42,~42\sqrt{AG}],$$
    and uses the algorithm from \cref{lem:svpSparcification2} to get a basis $B'$ of sparse sub-lattice lattice $\lat' \subseteq \lat$ of rank $\order{n}$, and outputs $(B', r)$. The efficiency follows from the efficiency of \cref{lem:AGRelation} and \cref{lem:svpSparcification2}. We now prove correctness. 
    \textbf{Correctness.} 
        We bound the probability of a lattice vector of a particular length making it into the sparse sub-lattice $\lat'$ lattice. 
        Suppose that the input was a $\yes$ instance of $\lin_\varepsilon$. Let $\mathcal{S}$ be the set of lattice vectors in $\lat$ of length at most $r$. 
        These are distinct lattice vectors such that their last coordinates are all $1$, and  $r < q \lambda_1^{(p)}(\lat)/2$ also holds\footnote{this is because $q$ grows exponentially in $m$, where as $r$ is $\order{m^{1/p}}$ }. Consider any two vectors $\vec v_1$ and $\vec v_2$ in $\mathcal{S}$. We claim that $B^{-1}(\vec v_1-\vec v_2) \neq 0 \mod q$. To see this, observe that if $B^{-1}(\vec v_1-\vec v_2) = 0 \mod q$, then $(\vec{v}_1-\vec v_2) \in q \cdot \cL(B) \implies \|(\vec{v}_1-\vec v_2)\|_p\geq q\lambda_1^{(p)}(\cL)$, which gives a contradiction by triangle inequality. This implies that the coefficient vectors of the lattice vectors in $\mathcal{S}$ are distinct and hence pairwise linearly independent in $\F_q$ because the last coordinate of coefficient vectors is also fixed to be 1. Then,
        $$\abs{\mathcal{S}} \geq G \implies \frac{q}{\abs{\mathcal{S}}} \leq {42} \cdot \sqrt{\frac{A}{G}} \in \order{-2^{m/2}}.$$
        Therefore by the lower bound in \cref{lem:svpSparcification2}, $(B', r)$ is a \yes instance of $\svp_{p, \gamma}$ except with an exponentially small probability. Next, suppose that the input was a \no instance of $\lin_\varepsilon$. Let $\mathcal{S}$ be the set of vectors of length {\em at most} $\gamma r$. We have that $\abs{\mathcal{S}} \leq A$, and hence the number of linearly independent vectors in $\mathcal{S}$ of length at most $\gamma r$ is at most $A$.
        Additionally, we claim that in this case $\gamma r \leq q \lambda_1^{(p)}(\lat)$. To see this, suppose $\gamma r > q \lambda_1^{(p)}(\lat)$ for contradiction. This implies that $\lambda_1^{(p)}(\lat) < \frac{\gamma r}{q} < \frac{\gamma r}{A}$, because $q > A$. Let $\vec v \in \lat$ be such that $\norm{\vec v}_p = \lambda_1^{(p)}(\lat)$. Then, $\{ -A \vec v, \ldots A\vec v \}$ are $2A$ vectors of length at most $\gamma r$, contradicting the fact that $\nump{\lat, \gamma r, \vec 0} \leq A$.  
        Therefore, 
        $$ \frac{\abs{\mathcal{S}}}{q} \leq 42 \cdot \sqrt{\frac{A}{G}} \in \order{2^{-m/2}},$$
        and by the upper bound in \cref{lem:svpSparcification2}, $(B', r)$ is a \no instance of $\svp_{p, \gamma}$ except with an exponentially small probability. \qedhere

\end{proof}

Together with \cref{cor:linEthHard}, we get the following theorem.

\thmsvpethhard*



\section{A reduction from \texorpdfstring{$\CVP_{p,\gamma'}$}{CVP p,gamma} to \texorpdfstring{$\bdd_{p,\alpha}$}{BDD p,alpha}}
\label{sec:bdd}
In this section we give  a randomized Karp reduction from $\cvp_{p, \gamma'}$ in a lattice in $m$ dimensions, for any constant $\gamma'>1$ to $\bdd_{p, \alpha}$ in a lattice in $\order{m}$ dimensions, for all $\alpha > \alpha_p^{\ddagger}$, where $\alpha_p^{\ddagger}$ is as defined in \cref{eqn:alphapddagger}.
We generally follow the reduction from \cite[Section 3]{bdd21}, but give simpler proofs that work for {\em any} $\cvp$ instance. 
Our reduction will use the following property of the integer lattice.

\begin{lemma}[\cite{bdd21}, Lemma 3.13]
\label{lem:intgadgetbdd}
For any $p \in [1,\infty)$, $\alpha_p^{\ddagger} < \alpha_G $ and $ \alpha_A < \alpha_G$, there exist $t \in [0,1/2]$, $C_r \geq t$ and $\phi_0, \phi_1 > 1$, such that for any $d\in \mathbb{Z}^+$ and for $\dgr r = C_r \cdot(d)^{1/p},~\vec t^{\dagger} = t \cdot \vec 1_{d}$,
\begin{equation}
\label{eqn:bddintgadget}
N_p(\mathbb{Z}^{d}, \alpha_G \cdot r^\dagger, \vec{t}^\dagger) \geq \max \left\{ \phi_0^{d - o(d)} \cdot N_p(\mathbb{Z}^{d}, r^\dagger , 0), \ 
\phi_1^{d - o(d)} \cdot N_p(\mathbb{Z}^{d}, \alpha_A   \cdot r^\dagger, \vec{t}^\dagger) \right\}.
\end{equation}
Furthermore\footnote{In \cite{bdd21} the right hand side appears in terms of $N_p^{\circ}$; our inspection of their proof shows that it works for $N_p$ as well.}, the constants $t, C_r, \phi_0, \phi_1$ only depends upon $p$, and can be efficiently computed from $p$.

\end{lemma}

\begin{lemma}
\label{lem:bddtransformation1}
    For any $p \in [1, \infty)$, suppose $(B', \vec t', r')$ is a $\cvp_{p, \gamma'}$ instance over an integer lattice such that $B \in \Z^{m \times m'}, \vec{t}' \in \Z^m$. Then for any $d \in \Z^+$, there exists an efficiently computable matrix-vector pair $(B, \vec t)$ such that $B$ generates a lattice $\lat(B)$ in $m + d$ dimensions, and for any radius $r > 0$ and constants $\alpha, s > 0$, 
    \begin{equation}
        N_p \left(\cL(B), r/\alpha,\vec{0}\right) \leq N_p \left(\Z^m, r/\alpha,\vec{0}\right) \cdot N_p \left(\cL^\dagger, \frac{r}{\alpha  s},\vec{0}\right), \label{eqn:bddineq1}
    \end{equation}
    and 
    \begin{itemize}
              \item  if   $\operatorname{dist}_p({\vec{t}'},\cL({B '})) \leq r'$ then,
   
   \begin{equation}
       N_p \left(\cL(B), r, \vec{t}\right) \geq N_p \left(\cL^\dagger, \frac{(r^{p} -r'^p)^{1/p}}{s},\vec{t}^\dagger\right). \label{eqn:bddineq2}
   \end{equation}
   
   \item if $\operatorname{dist}_p({\vec{t}'},\cL({B}')) \geq  \gamma'r'$ then,
   \begin{equation}
         N_p \left(\cL(B),r, \vec{t}\right) \leq N_p(\Z^m,r, \vec{0} ) \cdot N_p\left(\cL^\dagger, \frac{(r^p-(\gamma'r')^p)^{1/p}}{s}, \vec{t}^\dagger \right). \label{eqn:bddineq3}
   \end{equation} 
    \end{itemize}
\end{lemma}

\begin{proof}
Fix any $p, d, s$. Define $B^\dagger = I_{d}$, $\vec t^\dagger= t\cdot \vec 1_{d}$, and set
\begin{align*}
        B &:= \begin{pmatrix}
            {B'}  & 0 \\ 
            0    & s B^{\dagger}
        \end{pmatrix};
        \quad
        \vec t := \begin{pmatrix}
            {\vec{t}'} \\ s \vec t^{\dagger}
        \end{pmatrix}.
\end{align*}

The transformation is clearly efficient. We now show the inequalities. 
For any radius $r$ and constant $\alpha$, consider a vector $\vec{v} \in \ballp{\cL, r/\alpha,\vec{0}}$ such that $\vec v = B\cdot(\vec x, \vec y)$, for some $\vec x \in \Z^{m'}, \vec y \in \Z^{d}$. Then,
\begin{align*}
        \norm{B' \cdot \vec{x}}_p^p + \norm{sB^\dagger \cdot \vec{y}}_p^p &\leq (r/\alpha)^p. 
\end{align*}
This implies that $\norm{B' \cdot \vec{x}}_p\leq r/\alpha$ and $\norm{B^\dagger \cdot \vec{y}}_p\leq r/s\alpha$. 
Since $\cL(B') \subseteq \Z^m$ number of possible values of $B' \cdot \vec{x}$ is upper bounded by $N_p \left(\Z^m, r/\alpha,\vec{0}\right)$. Similarly, number of possible values of $B^\dagger \cdot \vec{y}$ is upper bounded by $N_p \left(\cL^\dagger, r/s \alpha,\vec{0}\right)$. Together, these imply that the number of possible vectors $\vec{v}$ is upper bounded by $N_p \left(\Z^m, r/\alpha,\vec{0}\right) \cdot N_p \left(\cL^\dagger, {r}{/\alpha  s},\vec{0}\right)$, which gives us \cref{eqn:bddineq1}.

Now suppose that the input instance was a \yes instance.
This implies that there exists $ \vec{x} \in \Z^{m'}$ such that $\norm{B' \cdot \vec{x} - \vec{t}'}_p^p \leq r'^p$. 
Let $\vec{u} \in \ballp{\cL^\dagger,\left(r^p -r'^p\right)/{s} ,\vec{t}^\dagger}$ be such that $\vec{u}=B^\dagger \cdot \vec{y}$. Then, 
$$\norm{sB^\dagger \cdot \vec{y} -s\vec{t}^\dagger}_p^p\leq \left(r^p -r'^p\right). $$ 
Then, $\vec v=(B' \vec{x},s \vec{u})$ is a vector in $\lat(B)$ such that $\normpp{\vec{v} - \vec{t}}\leq r^p$, which implies that $\vec v \in \ballp{\cL,r ,\vec{t}}$. This mapping from $\vec u$ to $\vec v$ is injective. Therefore, 
    \[  N_p \left(\cL, r, \vec{t}\right)  \geq N_p \left(\cL^\dagger, \frac{(r^{p} -r'^p)^{1/p}}{s},\vec{t}^\dagger\right),\]
    which gives us \cref{eqn:bddineq2}.

    Next, suppose that the input instance was a \no instance. Then for all $ \vec{x} \in \Z^{m'}$, $$\norm{B' \cdot \vec{x} - \vec{t}'}_p^p \geq (\gamma'r')^p.$$ 
    Let $\vec{v} \in B_p(\cL, r,\vec{t})$ such that $\vec{v}=(B' \cdot \vec{x}, sB^\dagger \cdot \vec{y})$ for some $\vec{x } \in \Z^{m'}, \vec{y} \in \Z^{d}$. This implies that
    \begin{align*}
        \norm{B' \cdot \vec x - \vec{t'}}_p &\leq r; \\ \norm{sB^\dagger \cdot \vec y -s\vec{t}^\dagger}_p &\leq \paren{r^p-(\gamma'r')^p}^{1/p}.
    \end{align*}
    Since $\cL' \subseteq \Z^m$ and $\vec{t}' \in \Z^m$, the number of possible values of $B' \cdot \vec{x}$ can take is upper bounded by $$N_p\paren{\cL',r,\vec{t}'}\leq N_p\paren{\Z^m,r,\vec{t}'}= N_p\paren{\Z^m,r,\vec{0}}.$$ 
    Similarly number of possible values for $sB^\dagger \vec{y}$ can take is upper bounded by $$N_p \paren{\cL^\dagger, \paren{r^p-(\gamma'r')^p}^{1/p}/{s}, \vec{t}^\dagger}.$$ By multiplying these upper bounds we get that, 
    \[N_p \left(\cL,r, \vec{t}\right) \leq N_p(\Z^m,r, \vec{0} ) \cdot N_p\left(\cL^\dagger, \frac{(r^p-(\gamma'r')^p)^{1/p}}{s}, \vec{t}^\dagger \right).\]
    This gives us \cref{eqn:bddineq3}.
\end{proof}

\begin{lemma}
\label{lem:bdd_main}
    For all $p \in [1, \infty)$, $\alpha > \alpha_p^{\ddagger}$, $c>0$ and $\gamma'>1$ the following holds for all sufficiently large $m \in \Z^+$. There exists an efficient algorithm that takes as an input a $\cvp_{p, \gamma'}$ instance $(B', \vec t', r')$, such that $B' \in \Z^{m \times m'}, \vec{t}' \in \Z^m$ and $r' = c {m^{1/p}}$, and returns $(B, \vec t, r, A, G)$, where $B$ generates a lattice in $\order{m}$ dimensions, and $A, G$ are integers such that $G > 2^m A$ and $r > 0$,
    \begin{itemize}
        \item  If $\di_p\paren{\vec{t}',\cL(B')}\leq r'$, then $\nump{\cL(B),r/\alpha, \vec{0}} \leq A$ and $\nump{\cL(B),r,\vec{t}}\geq G.$
       \item  If  $\di_p\paren{\vec{t}',\cL(B')}\geq \gamma'r'$, then $\nump{\cL(B), r,\vec{t}} \leq A.$  
    \end{itemize}
\end{lemma}

\begin{proof}
Fix any $p$. 
Define $B^\dagger = I_{d}$, $\vec t^\dagger= t\cdot \vec 1_{d}$. On input a $\cvp_{p, \gamma'}$ instance, the reduction applies the transformation from \cref{lem:bddtransformation1} for some choice of $s, d, t$ to be determined in the following, and receives $(B, \vec t)$. The reduction returns $(B, \vec t, r, A, G)$, where $r, A, G$ will also be determined in the following. The reduction is clearly efficient, as long as these parameters are efficiently computable. We now show correctness. 

We will first fix $r$. 
Let $\delta_1, \delta_2>0$, be arbitrarily small parameters that control the closeness of $\alpha$ to $\alpha_p^\ddagger$, such that 
\[ \alpha=(1+\delta_1)(1+\delta_2) \cdot \alpha_p^{\ddagger}. \]

Now define\footnote{We can assume $\delta_2$ to be small enough so that $\alpha_A$ is well defined and $\alpha_A>0$} 
\[\alpha_G := (1+\delta_1) \alpha_p^{\ddagger} ,\; \alpha_A := \left((1+\delta_2)^p-{\gamma'^p((1+\delta_2)^p-1)}\right)^{1/p}\alpha_G.\]

Set,
$$r^p=\left(\frac{(1+\delta_2)^p}{(1+\delta_2)^p-1}\right)\cdot r'^p.$$

\Cref{lem:intgadgetbdd} guarantees that for this choice of $\alpha_A,\alpha_G$ there exist $t,C_r,\phi_0$ and $\phi_1$ that satisfy \cref{eqn:bddintgadget}. 
Set $d=C m$, for a constant $C$ large enough so that the following two inequalities are satisfied.
\begin{align*}
     \phi_0^{d -o(d)} &\;\geq  2^m \cdot N_p\left(\Z^m,r/\alpha,0\right)\;,\\
     \phi_1^{d -o(d)} &\;\geq  2^m \cdot N_p\left(\Z^m,r,0\right)\;.\\
\end{align*}
Since $r=O(m)^{1/p}$, \cref{lemma:theta0} implies that there exists a constant $C$ such that the above two inequalities are satisfied. We will set $B^\dagger = I_{d} $ , $\vec t^\dagger =t \cdot \vec{1}_{d}$ and $r^\dagger= C_r (d)^{1/p}$, where $t$ and $ C_r$ are guaranteed by \cref{lem:intgadgetbdd}.
We now set 
\[G := N_p\left(\cL^\dagger,\alpha_G\cdot r^\dagger, \vec{t}^\dagger\right)\;,\]
and,
\[A:=\max\left\{N_p(\Z^m,r, \vec{0} )\cdot N_p\left(\cL^\dagger,\alpha_A\cdot r^\dagger, \vec{t}^\dagger\right) \;,\;N_p(\Z^m,r/\alpha, \vec{0} ) \cdot N_p\left(\cL^\dagger,r^\dagger, \vec{0}\right) \cdot \right\} \;.\]

Notice that both $A$ and $G$ can be approximated to a high precision in $\poly(m)$ time by using the theta function.

Now, in order to have that $G \geq 2^m A$, we will use \cref{lem:intgadgetbdd}. It is enough to satisfy the following inequalities.
\begin{align}
     \frac{r}{\alpha  s} & \;\leq r^\dagger, \label{eqn:bddweirdineq1} \\
     \frac{(r^{p} -r'^p)^{1/p}}{s} & \:\geq \alpha_G \cdot r^\dagger, \label{eqn:bddweirdineq2} \\
     \frac{(r^{p} - (\gamma'r')^p)^{1/p}}{s} &\;\leq \alpha_A \cdot r^\dagger. \label{eqn:bddweirdineq3}
\end{align}

First, we set $s$ such that the second inequality \cref{eqn:bddweirdineq2} is tight: $s= \dfrac{(r^p-r'^p)^{1/p}}{\alpha_G \cdot r^\dagger}$. 
For the first inequality \cref{eqn:bddweirdineq1} observe, 
\begin{align*}
    \frac{r}{\alpha s}&= \frac{r \cdot r^\dagger}{(1+\delta_2) \cdot(r^p-r'^p)^{1/p}}\;.\\
\end{align*}
We had set $r^p=\left(\frac{(1+\delta_2)^p}{(1+\delta_2)^p-1}\right)\cdot r'^p$, which implies that $\frac{r}{(r^p -r'^p)^{1/p}}=(1+\delta_2)$.
Therefore, $$\frac{r}{\alpha s} = r^\dagger \leq r^\dagger.$$
For the last inequality \cref{eqn:bddweirdineq3}, observe that \[ \frac{(r^p-(\gamma'r')^p)^{1/p}}{s \cdot \alpha_A}= \frac{(r^p -(\gamma'r')^p)^{1/p} \cdot \alpha_G \cdot r^\dagger}{(r^p -r'^p)^{1/p} \cdot \alpha_A} \;,\]
Note that,
\[\frac{(r^p-(\gamma'r')^p)^{1/p}}{(r^p -r'^p)^{1/p}} = \left((1+\delta_2)^p-\gamma'^p \cdot{\left((1+\delta_2)^p-1\right)}\right)^{1/p} \;,\]
which implies,
\[ \frac{(r^p-(\gamma'r')^p)^{1/p}}{s \cdot \alpha_A} = r^\dagger \leq r^\dagger\;.\]
Hence the third inequality is also satisfied.
This implies the following holds,
\begin{align*}
    N_p\left(\cL^\dagger,\frac{(r^p-r'^p)}{s}, \vec{t}^\dagger\right)&\;\geq\; N_p\left(\cL^\dagger,\alpha_G\cdot r^\dagger, \vec{t}^\dagger\right)\;,\\
     N_p\left(\cL^\dagger,\frac{(r^p- (\gamma'r')^p)}{s}, \vec{t}^\dagger\right)&\;\leq\;N_p\left(\cL^\dagger,\alpha_A\cdot r^\dagger, \vec{t}^\dagger\right),\\
     N_p\left(\cL^\dagger,\frac{r}{\alpha s}, \vec{0}\right)&\;\leq\;N_p\left(\cL^\dagger,r^\dagger, \vec{0}\right).
\end{align*}

By \cref{lem:intgadgetbdd} and our choice of $d$, it holds that $G > 2^m A$. 
\end{proof}



\thmcvpbddreduction*

\begin{proof}

     Given $B ' \in \Z^{m \times m'}$ and $\vec{t}' \in \Z^m$, the reduction calls the algorithm from \cref{lem:bdd_main} to get the corresponding $(B,\vec{t},r,A,G)$. It then uses the LLL algorithm \cite{LLL82} to compute a basis for $\lat(B)$. Let $\kappa$ denote the rank of $\cL(B)$, it is clear from the construction in \cref{lem:bdd_main} that $\kappa\geq m$. 
     Now the reduction generates a prime number $q\in [\sqrt{AG}/42,42\sqrt{AG}]$, and samples $\vec{x} \sim \F^{\kappa}$ and $\vec{z} \sim \F^{\kappa}$ uniformly at random. 
     Next, the reduction uses lattice sparsification \cref{ssec:sparcification_prelims} to find a sparser sub-lattice $\lat'' \subseteq \lat(B')$. Precisely, it sets
     \[\lat'' := \{ \vec v \in \mathcal{L}(B) :  \inner{B^+ \vec v, \vec x} \equiv 0 \pmod{q} \}, \quad \vec v'' := \vec t - B\vec z.\]
     Then the reduction uses its $\BDD_{p, \alpha}$ oracle with $(\lat'', \vec t'')$ as input, and outputs \yes if and only if it outputs a lattice vector $\vec v$ satisfying $\norm{\vec v - \vec t}_p \leq r$. Since the reduction uses only polynomial time algorithms, it is clearly efficient. We now show correctness.
     
     For a sufficiently large $m$, it holds that $r< q \cdot \lambda_1^{(p)}(\cL(B))/2$. 
     Suppose that the input was a \yes instance of $\cvp_{p, \gamma'}$. Then, $\dist_p\paren{\vec{t}',\cL(B')}\leq r'$. By the guarantee of \cref{lem:bdd_main}, $\nump{\cL(B),r,/\alpha, \vec{0}} \leq A$ and $\nump{\cL(B),r,\vec{t}}\geq G.$ Therefore, using \cref{lem:bdd_sparcification} we get that.
    \[ \Pr[\lambda_1(\mathcal{L}'') \leq r/\alpha] \leq \frac{\nump{\mathcal{L}(B) , r/\alpha, \vec 0}}{q}\leq \dfrac{A}{q} \leq 42 \sqrt{\dfrac{A}{G}} \leq \frac{42}{2^{m/2}}.\]
    Similarly, 
    \[\Pr[\operatorname{dist}_p({\vec t''}, {\mathcal{L}''}) > r] \leq \frac{q}{\nump{\mathcal{L}(B), r, \vec t}} + \frac{1}{q^\kappa}\leq \dfrac{q}{G}+\frac{1}{q^\kappa} \leq 42\sqrt{\frac{A}{G}} \;+ \frac{1}{q^\kappa} \leq \frac{43}{2^{m/2}}.\]
    
    By a union bound, with probability at least $\paren{1- 2^{-\Omega(m)}}$, $\dist_p(\lat'', \vec t') \leq r < \alpha \lambda_1(\lat'')$ and thus the pair $(\lat'', \vec t'')$ satisfies the $\bdd_{p, \alpha}$ promise. 
    In this case, the $\bdd_{p, \alpha}$ oracle outputs a vector $\vec v$ such that $\norm{\vec v - \vec t} \leq r$, and thus the reduction outputs $\yes$. 

Next, suppose that the input was a \no instance. Then,  $\dist_p\paren{\vec{t}',\cL(B')}> \gamma' r'$. By the guarantee of \cref{lem:bdd_main}, $\nump{\cL(B), r,\vec{t}} \leq A.$ Again by applying \cref{lem:bdd_sparcification} we get that,
         \[\Pr[\dist_p(\vec t'', \lat'') \leq r] \leq \frac{\nump{\mathcal{L}(B), r, \vec t}}{q} + \frac{1}{q^\kappa}\leq \frac{A}{q} + \frac{1}{q^\kappa}\leq 42\sqrt{\frac{A}{G}}+\frac{1}{q^\kappa} \leq \frac{43}{2^{m/2}}.\]
         
         In this case, with probability at least $1 - 2^{-\Omega(m)}$, there are no lattice vectors within a distance $r$ from $\vec t''$, and therefore the reduction outputs \no with at least as much probability. 
\end{proof}

Together with the reduction from $\lin_\varepsilon$ to $\cvp_{p, \gamma}$ in \cref{thm:LinCvpReduction}, we get the following corollary.
 
\begin{corollary}
\label{lintobdd}
        For all $p \in [1, \infty)$, $\alpha > \alpha_p^{\ddagger}$, and for all sufficiently large $m$, there exists a polynomial time randomized Karp reduction from $\lin_{\varepsilon}$ over $n$ variables and $m$ equations $\bdd_{p, \alpha}$ in rank $\order{m}$. 
    \end{corollary}

From \cref{cor:linEthHard} and \cref{lintobdd} we get the following,

\thmbddethhard*

\section*{Acknowledgments}

This work was supported by NRF grant NRF-NRFI09-0005. 

\bibliographystyle{alpha}
\bibliography{references}


\end{document}